\documentclass[11pt,a4paper]{article}
\usepackage[utf8]{inputenc}
\usepackage[english]{babel}
\usepackage{hyperref}
\usepackage{amsmath}
\usepackage{amsmath,amsthm,enumerate}
\usepackage{enumerate}
\usepackage{amssymb}
\usepackage{xspace}
\usepackage{tikz}
\usetikzlibrary{calc}
\usepackage{shapepar}
\usepackage{comment}
\usetikzlibrary{fit}
\usetikzlibrary{shapes.geometric}
\usepackage{lmodern}
\usetikzlibrary{decorations.pathmorphing}
\tikzset{snake it/.style={decorate, decoration=snake}}
\usepackage[color=green!30]{todonotes}	%

\usepackage{aliascnt}	

\usepackage[hmargin=2.5cm,vmargin=3cm]{geometry}

\newtheorem{theorem}{Theorem}

\newaliascnt{claim}{theorem}

\newtheorem{proposition}[theorem]{Proposition}
\newtheorem{lemma}[theorem]{Lemma}
\newtheorem{corollary}[theorem]{Corollary}

\newtheorem{claim}[claim]{Claim}
\newtheorem{foucClaim}[claim]{Claim}
\newtheorem{observation}[theorem]{Observation}

\aliascntresetthe{claim}

\newcommand{\fullversion}[1]{#1}
\newcommand{\confversion}[1]{}

\usepackage{etoolbox}

\newcommand{\appendixproof}[2]{
#2
}

\newcommand{\appendixDefinition}[1]{
#1
}

\newcommand{\COL}{\textsc{Coloring}\xspace}

\newcommand{\PBG}{\textsc{Grundy Coloring}\xspace}

\newcommand{\wGC}{\textsc{Weak Grundy Coloring}\xspace}
\newcommand{\cGC}{\textsc{Connected Grundy Coloring}\xspace}

\newcommand{\G}{\Gamma}
\newcommand{\cG}{\Gamma_c}
\newcommand{\Sat}{\textsc{SAT}\xspace}

\newcommand{\mc}{\mathcal}

\newcommand{\np}{\mathsf{NP}}
\newcommand{\paranp}{\mathsf{para}\textsc{-}\mathsf{NP}}
\newcommand{\fpt}{\mathsf{FPT}}
\newcommand{\xp}{\mathsf{XP}}
\newcommand{\wone}{\mathsf{W[1]}}

\renewcommand\leq\leqslant
\renewcommand\geq\geqslant

\newlength{\atextwidth}
\newcommand{\problemdec}[3]{
  \vspace{1mm}
\noindent\fbox{
  \begin{minipage}{\atextwidth}
  \begin{tabular*}{\textwidth}{@{\extracolsep{\fill}}lr} #1 \\ \end{tabular*}
  {\bf{Input:}} #2  \\
  {\bf{Question:}} #3
  \end{minipage}
  }
  \vspace{1mm}
}

\begin{document}

\title{Complexity of Grundy coloring and its variants}

\author{\'Edouard Bonnet\footnote{\noindent LAMSADE - CNRS UMR 7243, PSL, Universit\'e Paris-Dauphine, France. \{edouard.bonnet,eunjung.kim,florian.sikora\}@dauphine.fr}
\and Florent Foucaud\footnote{\noindent LIMOS - CNRS UMR 6158, Universit\'e Blaise Pascal, France. florent.foucaud@gmail.com}~~\footnotemark[1]
\and Eun Jung Kim\footnotemark[1]
\and Florian Sikora\footnotemark[1]
}

\maketitle

\begin{abstract}
The Grundy number of a graph is the maximum number of colors used by the greedy coloring algorithm over all vertex orderings. In this paper, we study the computational complexity of \PBG, the problem of determining whether a given graph has Grundy number at least~$k$. We also study the variants \wGC (where the coloring is not necessarily proper) and \cGC (where at each step of the greedy coloring algorithm, the subgraph induced by the colored vertices must be connected).

We show that \PBG can be solved in time $O^*(2.443^n)$ and \wGC in time $O^*(2.716^n)$ on graphs of order~$n$. While \PBG and \wGC are known to be solvable in time $O^*(2^{O(wk)})$ for graphs of treewidth~$w$ (where $k$ is the number of colors), we prove that under the Exponential Time Hypothesis (ETH), they cannot be solved in time $O^*(2^{o(w\log w)})$. We also describe an $O^*(2^{2^{O(k)}})$ algorithm for \wGC, which is therefore $\fpt$ for the parameter $k$. Moreover, under the ETH, we prove that such a running time is essentially optimal (this lower bound also holds for \PBG). Although we do not know whether \PBG is in $\fpt$, we show that this is the case for graphs belonging to a number of standard graph classes including chordal graphs, claw-free graphs, and graphs excluding a fixed minor. We also describe a quasi-polynomial time algorithm for \PBG and \wGC on apex-minor graphs. In stark contrast with the two other problems, we show that \cGC is $\np$-complete already for $k=7$ colors. 
%
%
\end{abstract}

\section{Introduction}

A \emph{$k$-coloring} of a graph $G$ is a surjective mapping $\varphi : V(G) \to
\{1,\ldots,k\}$ (we say that vertex $v$ is colored with $\varphi(v)$). A $k$-coloring $\varphi$  is \emph{proper} if any two adjacent vertices receive different
colors in $\varphi$. 
The \emph{chromatic number} $\chi(G)$ of $G$ is the smallest $k$ such that $G$ has a proper $k$-coloring. Determining the chromatic number of a graph is one of the most fundamental problems in graph theory.
Given a graph $G$ and an ordering $\sigma=v_1,\ldots,v_n$ of $V(G)$, the \emph{first-fit coloring algorithm} colors vertex $v_i$ with the smallest color of $\{v_1,\ldots,v_{i-1}\}$ that is not present among the set of neighbors of $v_i$. The
\emph{Grundy number} $\Gamma(G)$ is the largest $k$ such that $G$ admits a vertex ordering on which the first-fit algorithm yields a proper $k$-coloring. First-fit is presumably the simplest heuristic to compute a proper coloring of a graph. In this sense, the Grundy number gives an algorithmic upper bound on the performance of any heuristic for the chromatic number. This notion was first studied
by Grundy in 1939 in the context of digraphs and
games~\fullversion{\cite{Berge1973,Grundy39}}\confversion{\cite{Grundy39}}, and formally introduced 40 years later by Christen and Selkow~\cite{Christen}.
\fullversion{It was independently defined under the name \emph{ochromatic number} by Simmons~\cite{Simmons} (the two concepts were proved to be equivalent in~\cite{Erdos1987}).}
Many works have studied the first-fit algorithm in connection with on-line coloring algorithms, see for example~\cite{FirstFitInterval}.
A natural relaxation of this concept is the \emph{weak Grundy number}, introduced by Kierstead and Saoub~\cite{Kierstead2011}, where the obtained coloring is not asked to be proper. A more restricted concept is the one of \emph{connected Grundy number}, introduced by Benevides et al.~\cite{ConnectedGrundy}, where the algorithm is given an additional ``local'' restriction on the feasible vertex orderings that can be considered: at each step of the first-fit algorithm, the subgraph induced by the colored vertices must be connected.

The goal of this paper is to advance the study of the computational complexity of determining the Grundy number, the weak Grundy number and the connected Grundy number of a graph.

Let us introduce the problems formally. 
Let $G$ be a graph and let $\sigma=v_1,\ldots,v_n$ be an ordering of $V(G)$. 
A (not necessarily proper) $k$-coloring $\varphi : V(G) \to \{1,\ldots,k\}$ of $G$ is a \emph{first-fit coloring with respect to $\sigma$} if for every vertex $v_i$ and every color $c$ with $c<\varphi(v_i)$, $v_i$ has a neighbor $v_j$ with $\varphi(v_j)=c$ for some $j<i$. 
In particular, $\varphi(v_1)=1$. 
A vertex ordering $\sigma=v_1,\ldots,v_n$ is \emph{connected} if for every $i$, $1\leq i\leq n$, the subgraph induced by $\{v_1,\ldots, v_i\}$ is connected. 
A $k$-coloring $\varphi : V(G) \to \{1,\ldots,k\}$ is called a (i) \emph{weak Grundy}, (ii) {\em Grundy}, (iii) {\em connected Grundy coloring} of $G$, respectively, if it is a first-fit coloring with respect to some vertex ordering $\sigma$ such that (i) $\varphi$ and $\sigma$ has no restriction, (ii) $\varphi$ is proper, (iii) $\varphi$ is proper and $\sigma$ is connected, respectively.

The maximum number of colors used in a (weak, connected, respectively) Grundy coloring is called the ({\em weak, connected}, respectively) Grundy number and is denoted $\G(G)$ ($\G'(G)$ and $\cG(G)$, respectively). 
In this paper, we study the complexity of computing these invariants.

%
%
%

\problemdec{\PBG}{A graph $G$, an integer $k$.}{Do we have $\G(G)\geq k$?}

\confversion{The problems \wGC and \cGC are defined analogously.}
 
\fullversion{
\problemdec{\wGC}{A graph $G$, an integer $k$.}{Do we have $\G'(G)\geq k$?}

\problemdec{\cGC}{A graph $G$, an integer $k$.}{Do we have $\cG(G)\geq k$?}
}


Note that $\chi(G) \leqslant \G(G) \leqslant \Delta(G) + 1$, where $\chi(G)$ is the chromatic number and $\Delta(G)$ is the maximum degree of $G$. However, the difference $\G(G)-\chi(G)$ can be (arbitrarily) large, even for bipartite graphs. For example, the Grundy number of the tree of Figure~\ref{fig:tk} is~$4$, whereas its chromatic number is~$2$. Note that this is not the case for $\cG$ for bipartite graphs, since $\cG(G)\leq 2$ for any bipartite graph $G$~\cite{ConnectedGrundy}. However, the difference $\cG(G)-\chi(G)$ can be (arbitrarily) large even for planar graphs~\cite{ConnectedGrundy}.

\sloppy\paragraph{Previous results.} \PBG remains $\np$-complete on bipartite graphs~\cite{HSj} and their complements~\cite{Zaker-cobip} (and hence claw-free graphs and $P_5$-free graphs), on chordal graphs~\cite{thesis}, and on line graphs~\cite{HMY12}. Certain graph classes admit polynomial-time algorithms. There is a linear-time algorithm for \PBG on trees~\cite{HHB82}. This result was extended to graphs of bounded treewidth by Telle and Proskurowski~\cite{TP97}, who proposed a dynamic programming algorithm running in time $k^{O(w)}2^{O(wk)}n=O(n^{3w^2})$ for graphs of treewidth~$w$ (in other words, their algorithm is in $\fpt$ for parameter $k+w$ and in $\xp$ for parameter $w$).\footnote{The first running time is not explicitly stated in~\cite{TP97} but follows from their meta-theorem. The second one is deduced by the authors of~\cite{TP97} from the first one by upper-bounding $k$ by $w \log_2 n + 1$.} A polynomial-time algorithm for \PBG on $P_4$-laden graphs, which contains all cographs as a subfamily, was given in~\cite{AL12}.

Note that \PBG admits a polynomial-time algorithm when the number~$k$ of colors is fixed~\cite{Zaker}, in other words, it is in $\xp$ for parameter~$k$. 

\PBG has polynomial-time constant-factor approximation algorithms for inputs that are interval graphs~\cite{GL88,FirstFitInterval}, complements of chordal graphs~\cite{GL88}, complements of bipartite graphs~\cite{GL88} and bounded tolerance graphs~\cite{Kierstead2011}. However, there is a constant $c>1$ such that approximating \PBG within $c$ in polynomial time is impossible unless $\np$ $\subseteq$ $\mathsf{RP}$~\cite{Kortsarz} (a result extended to chordal graphs under the assumption $\mathsf{P}\neq\np$ in the unpublished manuscript~\cite{GV97}). It is not known whether a polynomial-time $o(n)$-factor approximation algorithm exists.

When parameterized by the graph's order minus the number of colors, \PBG was shown to be in $\fpt$ by Havet and Sempaio~\cite{HSj}.

\wGC was not studied as much as \PBG, but many results proved for \PBG also hold for \wGC. \wGC was shown to be $\np$-hard to approximate within some constant factor $c>1$, even on chordal graphs~\cite{GV97}. Furthermore, in~\cite{TP97} an algorithm for \wGC running in time $2^{O(wk)}n=O(n^{3w^2})$ for graphs of treewidth~$w$ was given (in~\cite{TP97}, \wGC was called \textsc{Iterated Dominating Set}).

\cGC was introduced by Benevides \emph{et al.}~\cite{ConnectedGrundy}, who proved it to be $\np$-complete, even for chordal graphs and for co-bipartite graphs. 

\paragraph{Our results.} We give two exact algorithms for \PBG and \wGC running in time $O^*(2.443^n)$ and $O^*(2.716^n)$, respectively. It was previously unknown if any $O^*(c^n)$-time algorithms exist for these problems (with $c$ a constant).
Denoting by $w$ the treewidth of the input graph, it is not clear whether the $O^*(2^{O(wk)})$-time algorithms for \PBG and \wGC of~\cite{TP97} can be improved, for example to algorithms of running time $O^*(k^{O(w)})$ or $O^*(f(w))$ (the notation $O^*$ neglects polynomial factors). In fact we show that an $O^*(k^{O(w)})$-time algorithm for \PBG would also have running time $O^*(2^{O(w\log w)})$.

 As a lower bound, we show that assuming the Exponential Time Hypothesis (ETH), an $O^*(2^{o(w\log w)})$-time algorithm for \PBG or \wGC does not exist (where $w$ is the feedback vertex set number of the input graph). In particular, the exponent $n$ cannot be replaced by the feedback vertex set number (or treewidth) in our $O^*(2.443^n)$-and $O^*(2.716^n)$-time algorithms.

We prove that on apex-minor-free graphs, quasi-polynomial time algorithms, of running time $n^{O(\log^2n)}$, exists for \PBG and \wGC.

We also show that \wGC can be solved in $\fpt$ time $O^*(2^{2^{O(k)}})$ using the color coding technique. Under the ETH, we show that this is essentially optimal: no $O^*(2^{2^{o(k)}}2^{o(n+m)})$-time algorithm for graphs with $n$ vertices and $m$ edges exists. The latter lower bound also holds for \PBG.

We also study the parameterized complexity of \PBG parameterized by the number of colors, showing that it is in $\fpt$ for graphs including chordal graphs, claw-free graphs, and graphs excluding a fixed minor.

Finally, we show that \cGC is computationally much harder than \PBG and \wGC when viewed through the lens of parameterized complexity. While for the parameter ``number of colors'', \PBG is in $\xp$ and \wGC is in $\fpt$, we show that \cGC is $\np$-complete even when $k=7$, that is, it does not belong to $\xp$ unless $\mathsf{P}=\np$. Note that the known $\np$-hardness proof of~\cite{ConnectedGrundy} for \cGC was only for an unbounded number of colors.

\paragraph{Structure of the paper.} We start with some preliminary definitions, observations and lemmas in Section~\ref{prelim}. Our positive algorithmic results are presented in Section~\ref{sec:positive}, and our algorithmic lower bounds are presented in Section~\ref{sec:negative}. We conclude the paper in Section~\ref{sec:conclu}.


\section{Preliminaries}\label{prelim}


\appendixDefinition
{
\noindent\paragraph{Graphs and sets.}
For any two integers $x<y$, we set $[x,y]:=\{x,x+1,\ldots,y-1,y\}$, and for any positive integer $x$, $[x]:=[1,x]$.
$V(G)$ denotes the set of vertices of a graph $G$ and $E(G)$ its set of edges.
For any $S \subseteq V(G)$, $E(S)$ denotes the subset of edges of $E(G)$ having both endpoints in $S$, and $G[S]$ denotes the subgraph of $G$ induced by $S$; that is, graph $(S,E(S))$.
If $H \subseteq V(G)$, $G-H$ denotes the graph $G[V(G) \setminus H]$.
As a slight abuse of notation, if $H$ is an (induced) subgraph of $G$, we also denote by $G-H$ the graph $G[V(G) \setminus V(H)]$.
For any vertex $v \in V(G)$, $N(v):=\{w \in V(G) | vw \in E(G)\}$ denotes the set of neighbors of $v$ in $G$.
For any subset $S \subseteq V(G)$, $N(S)=\bigcup_{v \in S} N(v) \setminus S$.
The \emph{distance-$k$} neighborhood of $v$ is the set of vertices at distance at most $k$ from $v$.

\paragraph{Computational complexity.}
A decision problem  is said to be \textit{fixed-parameter tractable} (or in the class $\fpt$) w.r.t. parameter $k$ if it can be solved in time $f(k)\cdot|I|^c$ for an instance $I$, where $f$ is a computable function and $c$ is a constant (see for example the books~\cite{DF99,Nie06} for details).
The class $\xp$ contains those problems solvable in time $|I|^{f(k)}$, where $f$ is a computable function.

The \emph{Exponential Time Hypothesis} (ETH) is a conjecture by Impagliazzo et al. asserting that there is no $2^{o(n)}$-time algorithm for \textsc{3-SAT} on instances with $n$ variables \cite{ImpagliazzoETH}. 
The ETH, together with the sparsification lemma \cite{ImpagliazzoETH}, even implies that there is no $2^{o(n+m)}$-time algorithm solving \textsc{3-SAT}.
Many algorithmic lower bounds have been proved under the ETH, see for example~\cite{LokshtanovSODA}.


\paragraph{Minors.} A \emph{minor} of a graph $G$ is a graph that can be obtained from $G$ by (i)~deletion of vertices or edges (ii)~contraction of edges (removing an edge and merging its endpoints into one). 
Given a graph $H$, a graph $G$ is \emph{$H$-minor-free} if $H$ is not a minor of $G$.

An \emph{apex graph} is a graph obtained from a planar graph $G$ and a single vertex $v$, and by adding arbitrary edges between $v$ and $G$. 
A graph is said to be \emph{apex-minor-free} if it is $H$-minor-free for some apex graph $H$.

\paragraph{Tree-decompositions.} A \emph{tree-decomposition} of a graph $G$ is a pair $(\mc{T},\mc{X})$, where $\mc{T}$ is a tree and $\mc{X}:=\{X_t:t\in V(\mc{T})\}$ is a collection of subsets of $V(G)$ (called \emph{bags}), and they must satisfy the following conditions: (i)~$\bigcup_{X\in V(T)}=V(G)$, (ii)~for every edge $uv\in E(G)$, there is a bag of $\mc{T}$ that contains both $u$ and $v$, and (iii)~for every vertex $v\in V(G)$, the set of bags containing $v$ induces a connected subtree of $\mc{T}$.

The maximum size of a bag $X_t$ over all tree nodes $t$ of $\mc{T}$ minus one is called the \emph{width} of $\mc{T}$. 
The minimum width of a tree-decomposition of $G$ is the \emph{treewidth} of $G$. 
The notion of tree-decomposition has been used extensively in algorithm design, especially via dynamic programming on the tree-decomposition.


}
\confversion{ We defer many (classic) technical definitions to the full version~\cite{Bonnet2014}, and only give the ones related to Grundy colorings.}\fullversion{
\noindent\paragraph{Grundy coloring.}} Let $G$ be a graph, $\sigma=v_1,\ldots,v_n$ an ordering of its vertices, $\varphi: V(G) \rightarrow [k]$ the proper first-fit coloring of $G$ with respect to $\sigma$, and $j$ the smallest index such that $\varphi(v_j)=k$.
Informally, finishing the Grundy coloring of $v_{j+1},\ldots,v_n$ is irrelevant in asserting that $\G(G) \geqslant k$, for this fact is established as soon as we color vertex $v_j$.
We formalize this idea that a potentially much smaller induced subgraph of the input graph (here, $G[\{v_1, \ldots, v_j\}]$) might be a relevant certificate, via the notion of \emph{witnesses} and \emph{minimal witnesses}\fullversion{\footnote{Witnesses were called \emph{atoms} by Zaker~\cite{Zaker}.}}.   

In a graph $G$, a \emph{witness achieving color $k$}, or simply a \emph{$k$-witness}, is an induced subgraph $G'$ of $G$, such that $\G(G') \geqslant k$.
Such a $k$-witness is \emph{minimal} if no proper induced subgraph of $G'$ has Grundy number at least $k$.

\begin{observation}\label{obs:minimal}
For any graph $G$,  $\G(G)\geq k$ if and only if $G$ admits a minimal $k$-witness.
\end{observation}

We can also notice that, in any Grundy $k$-coloring (that is, Grundy coloring achieving color $k$) of a minimal $k$-witness, exactly one vertex is colored with $k$.
Otherwise, it would contradict the minimality.

If $k$ is not specified, we assume that the witness achieves the largest possible color: a (minimal) \emph{witness} is a (minimal) witness achieving color $\G(G)$.
A \emph{colored} (minimal) ($k$-)witness is a (minimal) ($k$-)witness together with a Grundy $k$-coloring of its vertices, that can be given equivalently by the coloring function $\varphi$, or the ordering $\sigma$, or a partition $W_1 \uplus \ldots \uplus W_k$ of the vertices into color classes (namely, the vertices of $W_i$ are colored with $i$).

We will now observe that minimal $k$-witnesses have at most $2^{k-1}$ vertices.
To that end, we define a family of rooted trees sometimes called \emph{binomial trees}.
If, for each $i \in [l]$, $t_i$ is a tree rooted at $v_i$, $v[t_1,\ldots,t_l]$ denotes the tree rooted at $v$ obtained by adding $v$ to the disjoint union of the $t_i$'s and linking it to all the $v_i$'s. 
Then, the \emph{$i$-th child} of $v$ is $v_i$ and is denoted by $v(i)$.
We say that $v$ is the \emph{parent} of $v_i$.
We may also say that $v$ is the \emph{parent} of the tree $t_i$.
The set of binomial trees $(T_k)_{k \geqslant 1}$ is a family of rooted trees defined as follows (see Figure~\ref{fig:tk} for an illustration):
\begin{itemize}
\item $T_1$ consists only of one vertex (incidentally the root), and
\item $\forall k \geqslant 1$, $T_{k+1}=v[T_1,T_2,\ldots,T_k]$.
\end{itemize}

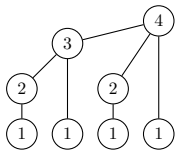
\begin{figure}[h!]
\centering
\begin{tikzpicture}[scale=0.6,transform shape]
\node[draw,circle] (r4) at (1,-0.5) {4};
\node[draw,circle] (r3) at (-1,-1) {3};
\node[draw,circle] (r32) at (-2,-2) {2};
\node[draw,circle] (r321) at (-2,-3) {1};
\node[draw,circle] (r31) at (-1,-3) {1};
\node[draw,circle] (r2) at (0,-2) {2};
\node[draw,circle] (r21) at (0,-3) {1};
\node[draw,circle] (r1) at (1,-3) {1};

\draw (r4) -- (r3);
\draw (r4) -- (r2);
\draw (r4) -- (r1);
\draw (r3) -- (r32);
\draw (r3) -- (r31);
\draw (r32) -- (r321);
\draw (r2) -- (r21);

\end{tikzpicture}
\caption{The binomial tree $T_4$, where numbers denote the color of each vertex in a first-fit proper coloring with largest number of colors.}\label{fig:tk}
\end{figure}

The binomial tree $T_k$ can be seen as the dependencies between the vertices of a minimal $k$-witness colored by a coloring $\varphi$.
More concretely, any vertex colored with color $i \leqslant k$ needs to have in its neighborhood $i-1$ vertices colored with each color from $1$ to $i-1$.
Say, we label the root of $T_k$ with the unique vertex colored~$k$. 
And then, in a top-town manner, we label, for each $j \in [\varphi(v)-1]$, the $j$-th child of a vertex labeled by $v$, by a neighbor of $v$ colored with $j$.
Each vertex of the minimal $k$-witness should appear at least once as a label of $T_k$, for the sake of minimality.
Besides, the number of vertex of $T_k$ is $2^{k-1}$.
This leads to the following observations:

\begin{observation}\label{obs:basic-props}
A minimal $k$-witness $W$ has radius at most $k$.
More precisely, $W$ is entirely included in the distance-$k$ neighborhood of the vertex colored with $k$ in a Grundy $k$-coloring of $W$.
\end{observation}

\begin{observation}\label{obs:basic-props2}
A minimal $k$-witness has at most $2^{k-1}$ vertices.
\end{observation}

\begin{observation}\label{obs:deg}
The color of a vertex of degree $d$ in any Grundy coloring is at most $d+1$.
\end{observation}

By Observations~\ref{obs:minimal}~and~\ref{obs:basic-props2}, \PBG can be solved by checking if one of the ${n \choose 2^{k-1}}$ induced subgraphs on $2^{k-1}$ vertices, has Grundy number $k$. 
This shows that, parameterized by the number $k$ of colors, the problem is in $\xp$:

\begin{corollary}[\fullversion{Zaker }\cite{Zaker}]\label{cor:gr-xp}
\PBG can be solved in time $f(k)n^{2^{k-1}}$. 
\end{corollary}


We now come back to binomial trees and show two lemmas that will be very helpful to prove the hardness results of the paper.

\begin{lemma}\label{lem:tree-coloring}
The Grundy number of $T_k$ is $k$. 
Moreover, there are exactly two Grundy colorings achieving color~$k$, and a unique Grundy coloring if we impose that the root $v$ is colored~$k$.
\end{lemma}
\appendixproof{Lemma~\ref{lem:tree-coloring}}
{
\begin{proof}
As hinted before, the tree $T_k$ is a minimal $k$-witness with the largest number of vertices, so $\G(T_k) \geqslant k$.
The easiest Grundy $k$-coloring of $T_k$ consists of coloring all the leaves with color $1$.
Now, if one removes all the leaves of $T_k$, one gets a binomial tree $T_{k-1}$, whose leaves can all be colored~$2$, and so forth, up to coloring $v$ with color $k$ (see~Figure~\ref{fig:tk}).
As the degree of $T_k$ is $k-1$, by Observation~\ref{obs:deg}, $\G(T_k) \leqslant k$ also holds. 

What remains to be seen is that the Grundy $k$-coloring of $T_k$ is unique up to deciding which of $v$ and $v(k-1)$ gets color $k$ and which gets color $k-1$.
There are only two vertices of degree $k-1$ in $T_k$: $v$ and $v(k-1)$.
Therefore, only $v$ or $v(k-1)$ can potentially be colored with $k$, by Observation~\ref{obs:deg}.  
As $T_k$ rooted at $v$ is isomorphic to $T_k$ rooted at $v(k-1)$, we can assume that $v$ will be the vertex colored $k$.  
We show by strong induction that there is only one Grundy coloring of $T_k$ where the root has color~$k$.
Obviously, there is a unique Grundy coloring of $T_1$.  
For any integer $k \geqslant 2$, if we impose that the root $v$ is colored $k$, the $k-1$ children of $v$ have to be colored with all the integers of $[k-1]$.  
As for each $i \in [k-1]$, $T_i$ has maximum degree $i-1$, the color of $v(i)$ is at most~$i$. 
First, color $k-1$ can only come from $v(k-1)$.
But, now that the color of this vertex is imposed, color $k-2$ can only come from $v(k-2)$. 
Finally, the only possibility is to color $v(i)$ with color $i$ for each $i \in [k-1]$. 
By the induction hypothesis, there is a unique such Grundy coloring for each subtree.
\end{proof}
}

\paragraph{Subtrees and dominant subtrees.} The subtree $t[x]$ rooted at vertex $x$ of a tree $t$ rooted at $v$, is the tree induced by all the vertices $y$ of $t$ such that the simple path from $v$ to $y$ goes through $x$.
The rooted tree $t'$ is a subtree of $t$, if there exists a vertex $x$ of $t$ such that $t'=t[x]$.
The \emph{number} of rooted subtree $t'$ of a rooted tree $t$ is the number of vertices $x$ of $t$ such that $t'=t[x]$.
In a binomial tree $T_k$, the number of $T_l$ (for $l \in [k-1]$) is $2^{k-l-1}$.
For any $l \in [k-1]$, we say that a subtree $T_l$ of $T_k$ is \emph{dominant}, if its root is the child of the root of a $T_{l+1}$.
In other words, a dominant subtree is the largest among its siblings.
The \emph{dominant subtree} of a vertex of a binomial tree is the largest subtree rooted at one of its children.
In a binomial tree $T_k$, the number of dominant $T_l$ (for $l \in [k-2]$) is by definition the number of $T_{l+1}$, that is $2^{k-l-2}$.

Although the statement of the next lemma is rather technical, its underlying idea is fairly simple.
If one removes some well-chosen subtree $T_{a_i}$ from a binomial tree $T_s$ rooted at $v$, and connects the parent $f$ of this removed subtree to the rest of a graph $G$, then in order to color $v$ with color $s$, one would have to color with $a_i$ at least one of the neighbors of $f$ outside $T_s$ (see Figure~\ref{fig:tk-partial}). Using this as a gadget, we will be able to make sure that at least one vertex of a specific vertex-subset is colored with a specific color.
We prove the more general result when multiple subtrees are removed.

\begin{lemma}\label{lem:tree-coloring-partial}
Let $a_1, \ldots, a_p < s$ be integers. Let $G$ be a graph, and let $T$ be an induced subgraph of $G$ such that the following hold.
\begin{itemize}
\item $T$ can be obtained from $T_s$ (rooted at $v$), a set $T_{a_1},\ldots, T_{a_p}$ of pairwise disjoint dominant subtrees. Let $F=\{f_1,\ldots,f_p\}$ be the set of parents of those subtrees.
\item We have $N(V(G) \setminus V(T)) \cap V(T) = F$, that is, only $F$ links $T$ to the rest of $G$.
\end{itemize}
Then, the following conditions on Grundy colorings of $G$ are equivalent.
\begin{itemize}
\item(i) There is a Grundy coloring such that $v$ is colored $s$. 
\item(ii) There is a Grundy coloring of an induced subgraph of $G-T$ such that, for each $i \in [p]$, at least one vertex of $N(f_i) \setminus V(T)$ is colored $a_i$ and no vertex of $N(f_i) \setminus V(T)$ is colored $a_i+1$.  
\end{itemize}
\end{lemma}
{
\begin{proof}
(ii) $\Rightarrow$ (i). 
Assume that there is a Grundy coloring of an induced subgraph of $G-T$ such that, for each $i \in [p]$, at least one vertex of $N(f_i) \setminus V(T)$ is colored $a_i$ and no vertex of $N(f_i) \setminus V(T)$ is colored $a_i+1$.  
We extend this Grundy coloring by coloring $T$ as we would optimally color $T_s$. 
By Lemma~\ref{lem:tree-coloring}, vertex $v$ will be colored with $s$.

(i) $\Rightarrow$ (ii). 
Now, suppose that there is a Grundy coloring where vertex $v$ receives color $s$.
By the same induction as in the second part of the proof of Lemma~\ref{lem:tree-coloring}, each vertex $f_i$ has to be colored $a_i+1$.
The degree of the neighbors of $f_i$ within $T$ is bounded by $a_i$, hence they cannot be colored with color $a_i$ (unless one first colors $f_i$ with a smaller color, but this would be a contradiction).
Thus, the color $a_i$ in the neighborhood of $f_i$ has to come from a vertex of $G-T$.
Moreover, no vertex in $N(f_i) \setminus V(T)$ can be colored $a_{i+1}$, otherwise $f_i$ cannot get this color.
Summing up, there is a Grundy coloring of an induced subgraph of $G-F$ such that, for each $i \in [p]$, at least one vertex of $N(f_i) \setminus V(T)$ is colored $a_i$ and no vertex of $N(f_i) \setminus V(T)$ is colored $a_i+1$.  
As $F$ separates $T-F$ from $G-T$, coloring vertices in $T-F$ is not helpful to color vertices in $N(f_i) \setminus V(T)$.
Hence, there is a Grundy coloring of an induced subgraph of {$G-\textbf{T}$} such that, for each $i \in [p]$, at least one vertex of $N(f_i) \setminus V(T)$ is colored $a_i$ and no vertex of $N(f_i) \setminus V(T)$ is colored $a_i+1$.
\end{proof}
}

\begin{figure}[h!]
\centering
\begin{tikzpicture}[scale=0.8,transform shape]
\tikzstyle{vertex}=[draw,circle,minimum size=12pt,inner sep=0pt]
\node[vertex] (r4) at (1,-0.5) {4};
\node[vertex] (r3) at (-1,-1) {3};
\node[vertex] (r31) at (-1,-3) {1};
\node[vertex] (r2) at (0,-2) {2};
\node[vertex] (r21) at (0,-3) {1};
\node[vertex] (r1) at (1,-3) {1};

\node[vertex] (e1) at (-3,-3) {};
\node[vertex] (e2) at (-3,-2) {2};
\node[vertex] (e3) at (-3,-1) {};
\node[vertex] (e4) at (-3,0) {};
\node[vertex] (e5) at (-3,1) {};

\node (f1) at (-1,-0.5) {$f_1$};
\node (v) at (1,-0.1) {$v$};
\node (nf1) at (-3,1.6) {$N(f_1) \setminus V(T)$};

\node (inv) at (-7,-1) {} ;

\node[draw, rectangle, rounded corners, fit=(e1) (e2) (e3) (e4) (e5)] (nei) {};
\node[draw, rectangle, rounded corners, fit=(inv) (nei)] (g) {$G-T$};

\foreach \i in {1,...,5}{
\draw (r3) -- (e\i) ;
}

\draw (r4) -- (r3);
\draw (r4) -- (r2);
\draw (r4) -- (r1);
\draw (r3) -- (r31);
\draw (r2) -- (r21);

\end{tikzpicture}
\caption{A simple instantiation of Lemma~\ref{lem:tree-coloring-partial} with $s=4$, $p=1$, and $a_1=2$.}\label{fig:tk-partial}
\end{figure}
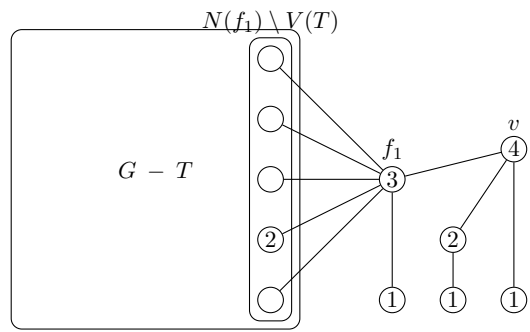


\paragraph{Weak Grundy and connected Grundy colorings.}
We can naturally extend the notion of witnesses to the \wGC problem.
It turns out that everything we observed or showed so far for \PBG, namely Observations~\ref{obs:minimal}, \ref{obs:basic-props}, \ref{obs:basic-props2}, \ref{obs:deg}, and Lemma~\ref{lem:tree-coloring} and \ref{lem:tree-coloring-partial} (where condition (ii) is replaced by the simpler condition: \emph{(ii') There is a weak Grundy coloring of an induced subgraph of $G-T$ such that, for each $i \in [p]$, at least one vertex of $N(f_i)$ is colored $a_i$}), are also valid when it comes to weak Grundy colorings.

For \cGC, again, we can similarly define a notion of witness. 
Though, as we will see, the size of minimal $k$-witnesses for the connected version cannot be bounded by a function of $k$.
Here, the only statements that remain valid are Observations~\ref{obs:minimal} and \ref{obs:deg}.
To illustrate the different behavior of this variant, the connected Grundy number of any binomial tree $T_k$ is $2$, as it is for every bipartite graph with at least one edge~\cite{ConnectedGrundy}.

\section{Positive results}\label{sec:positive}

We now present the positive algorithmic results of this paper.

\subsection{Exact algorithms for \PBG and \wGC}

A straightforward way to solve \PBG{} is to enumerate all possible orderings of the vertex set and to check whether the greedy algorithm uses at least $k$ colors. 
This is a  $\Theta(n!)$-time algorithm. 
A natural question is whether there is a faster exact algorithm. 
\fullversion{Such algorithms for \COL based on dynamic programming have been long known, see for example Lawler~\cite{Lawler1976}, but no $c^n$ algorithm for \PBG, for any constant $c$, was previously known.} 
We now give such an algorithm.


As a preparatory lemma, we remark that a colored minimal $k$-witness can be seen as a set of nested independent dominating sets, in the following sense.
\begin{lemma}\label{lem:independent-dominating}
Let $G$ be a graph and let $G'$ be a colored $k$-witness with the partition into color classes $W_1\uplus \cdots \uplus W_{k}$. 
Then, $W_i$ is an independent set which dominates the set $\bigcup_{j \in [i+1,k]} W_j$. 
In particular, $W_1$ is an independent dominating set of $V(G')$.
\end{lemma}
\appendixproof{Lemma~\ref{lem:independent-dominating}}
{
\begin{proof}
As a Grundy coloring is a proper coloring, $W_i$ is an independent set.  
If a vertex $v \in W_h$ (with $h>i$) has no neighbor in $W_i$, then $v$ is colored with a color at most $i$, a contradiction. 
So, $W_i$ should dominate $W_h$.
\end{proof}
}

We rely on two observations: (a) in a colored witness, every color class $W_i$ is an independent dominating set in $G[\bigcup_{j\geqslant i} W_j]$ (Lemma~\ref{lem:independent-dominating}), and (b) any independent dominating set is a maximal independent set (and vice versa). The algorithm is obtained by dynamic programming over subsets, and uses an algorithm which enumerates all maximal independent sets. 

\begin{theorem}\label{thm:exact}
\PBG can be solved in time $O^*(2.4423^n)$.
\end{theorem}
{
\begin{proof}
Let $G=(V,E)$ be a graph. We present a dynamic programming algorithm to compute $\G(G)$. For simplicity, given $S\subseteq V$, we denote the Grundy number of the induced subgraph $G[S]$ by $\G(S)$. We recursively fill a table $\G^*(S)$ over the subset lattice $(2^V,\subseteq)$ of $V$ in a bottom-up manner starting from $S=\emptyset$. The base case of the recursion is $\G^*(\emptyset)=0$. The recursive formula is given as 
$$\G^*(S)=\max \{\G^*(S\setminus X)+1~|~X \subseteq S \text{ is an independent dominating set of } G[S]\}.$$


Now let us show by induction on $|S|$ that $\G^*(S)=\G(S)$ for all $S\subseteq V$. The assertion trivially holds for the base case. Consider a nonempty subset $S\subseteq V$; by induction hypothesis, $\G^*(S')=\G(S')$ for all $S'\subset S$.
Let $X$ be a subset of $S$ achieving $\G^*(S)=\G^*(S\setminus X) +1$ and $X'$ be the set of the color class~1 in the ordering achieving the Grundy number $\G(S)$. 

Let us first see that $\G^*(S)\leq \G(S)$. By induction hypothesis we have $\G^*(S\setminus X)=\G(S\setminus X)$. Consider a vertex ordering $\sigma$ on $S\setminus X$ achieving $\G(S\setminus X)$. Augmenting $\sigma$ by placing all vertices of $X$ at the beginning of the sequence yields a (set of) vertex ordering(s). Since $X$ is an independent set, the first-fit algorithm gives color~1 to all vertices in $X$, and since $X$ is also a dominating set for $S\setminus X$, no vertex of $S\setminus X$ receives color~1. Therefore, the first-fit algorithm on such ordering uses $\G(S\setminus X)+1$ colors. We deduce that $\G(S)\geq \G(S\setminus X)+1=\G^*(S\setminus X)+1=\G^*(S)$.

To see that $\G^*(S)\geq \G(S)$, we first observe that $\G(S\setminus X') \geq \G(S)-1$. Indeed, the use of the optimal ordering of $S$ ignoring vertices of $X'$ on $S\setminus X'$ yields the color $\G(S)-1$. 
We deduce that $\G(S)\leq \G(S\setminus X')+1=\G^*(S\setminus X')+1\leq \G^*(S\setminus X)+1=\G^*(S)$.

As an independent dominating set is a maximal independent set, we can estimate the computation of the table by restricting $X$ to the family of maximal independent sets of $G[S]$. 
On an $n$-vertex graph, one can enumerate all maximal independent sets in time $O(1.4423^n)$~\cite{MoonMoser}. 
Thus, filling the table by increasing size of set $S$ takes:
$$\sum_{i=0}^n {n \choose i}\cdot  1.4423^i=(1+1.4423)^n.$$\end{proof}
%
}

A similar dynamic programming gives a slightly worse running time for \wGC.

\begin{theorem}\label{thm:exact-weak}
\wGC can be solved in time $O^*(2.7159^n)$.
\end{theorem}

\begin{proof} 
Now, we fill the table: $$\G^w(S)=\max \{\G^w(S\setminus X)+1~|~X \subseteq S \text{ is a minimal dominating set of } G[S]\}.$$
In a colored witness $W_1\uplus \cdots \uplus W_{k}$ of \wGC, for any $i \in [k]$, $W_i$ (is no longer necessarily an independent set and) dominates $\bigcup_{j \in [i+1,k]} W_j$.
To establish that, for any $S \subseteq V$, $\G^w(S)=\G'(S)$, we need to transform any colored witness $W_1 \uplus \cdots \uplus W_{k}$ (with $k \leqslant 2$) into a colored witness $W'_1 \uplus \cdots \uplus W'_{k}$ on the same induced subgraph $G'$, also achieving color $k$, but with the additional property that $W'_1$ is a minimal dominating set of $G'$.
Actually, in order to obtain that property we only need to transfer some vertices of $W_1$ to $W_2$.
We can choose $W'_1 \subseteq W_1$ to be any minimal dominating set of $G'$.
Then, we set $W'_2=W_2 \cup (W_1 \setminus W'_1)$.
For any $i \in [3,k]$, we just set $W'_i=W_i$.
As $W'_1$ is a dominating set of $G'$, the partition $W'_1 \uplus \cdots \uplus W'_{k}$ is indeed a colored witness.
Enumerating all the minimal dominating sets of a graph on $i$ vertices can be done in time $O^*(1.7159^i)$ \cite{FominGPS08}, hence the running time of our algorithm.
\end{proof}

We leave it as an open question to improve the running time of those algorithms.
We note that the \emph{fast subset convolution} technique~\cite{BjorklundHKK07}, which is commonly used to design exponential-time algorithms, does not seem to be directly applicable here.

\subsection{Quasi-polynomial algorithms for \PBG and \wGC on apex-minor-free graphs}

We will now show that the $\xp$ algorithms of~\cite{TP97} for \PBG and \wGC imply the existence of quasi-polynomial-time algorithms for these problems on apex-minor-free graphs (such as planar graphs).

The following result of Chang and Hsu~\cite{Chang} will be used:


\begin{theorem}[\cite{Chang}]\label{th:sparse-logn}
Let $G$ be a graph on $n$ vertices for which every subgraph $H$ has at most $d|V(H)|$ edges. Then $\G(G)\leq\log_{d+1/d}(n)+2$.
\end{theorem}

In fact, we note that the bound of Theorem~\ref{th:sparse-logn} also holds for the weak Grundy number, indeed the proof of~\cite{Chang} is still valid for this case.

A class of graphs has \emph{bounded local treewidth} if for any of its members $G$, the treewidth of $G$ is upper-bounded by a function of the diameter of $G$. 
The following result was proved by Demaine and Hajiaghayi~\cite{DemaineSODA}:

\begin{theorem}[\cite{DemaineSODA}]\label{thm:bded-LTW}
For every apex graph $H$, the class of $H$-minor-free graphs has
bounded local treewidth. More precisely, there is a function $f$ such that any $H$-minor-free graph $G$ of diameter $D$ has treewidth at most $f(H)D$.
\end{theorem}

\fullversion{In fact, it was proved by Eppstein~\cite{Eppstein} that a graph has
bounded local treewidth if and only if it is apex-minor-free.}

\begin{theorem}\label{thm:planar}
\PBG  and \wGC can be solved in time $n^{O(\log^2 n)}$ on apex-minor-free graphs of order $n$.
\end{theorem}
\appendixproof{Proposition~\ref{thm:planar}}
{
\begin{proof}
We first consider \PBG. Any $H$-minor-free graph of order $n$ has at most $f(H)n$ edges~\cite{Mader67} for some function $f$; hence, by Theorem~\ref{th:sparse-logn}, we have $k\leqslant\G(G)\leq c\log n$ for some constant $c$ (otherwise, we have a NO-instance).  
As noted in Observation~\ref{obs:basic-props}, any minimal $k$-witness is included in some distance-$k$ neighborhood of $G$. 
Hence, we apply the $O(n^{3w^2})$-time algorithm of~\cite{TP97} for graphs of treewidth at most $w$: for every vertex $v$ of $G$, apply it to the distance-$k$ neighborhood of $v$. 
This is a subgraph of diameter at most $2k=O(\log n)$, and by Theorem~\ref{thm:bded-LTW} it has treewidth $w=O(\log n)$ as well. 
Hence $O(n^{3w^2})=n^{O(\log^2 n)}$.

The same argumentation also works for \wGC. Indeed, as pointed out before, the bound of Theorem~\ref{th:sparse-logn} also holds for the weak Grundy number. Moreover, there is also an algorithm running in time $O(n^{3w^2})$ for \wGC~\cite{TP97} (where the problem is called \textsc{Iterated Dominating Removal}).
\end{proof}
}

In the light of Theorem~\ref{thm:planar}, it is natural to ask whether \PBG can be solved in polynomial time on apex-minor-free graphs (or planar graphs)? Note that by Theorem~\ref{thm:planar}, an NP-hardness result for \PBG on apex-minor-free graphs would contradict the ETH.

\subsection{\wGC parameterized by $k$ is in $\fpt$}

We recall that \wGC is $\np$-complete~\cite{GV97}. 
In this subsection, we show that \wGC has an $O^*(2^{2^{O(k)}})$-time algorithm (Theorem~\ref{thm:weakFPT}). We will later show that this running time is essentially optimal under the ETH (Theorem~\ref{thm:weakFPT-lowerbound}).

\begin{theorem}\label{thm:weakFPT}
\wGC can be solved in time $O^*(2^{2^{O(k)}})$, where $k$ is the number of colors.
\end{theorem}
\confversion{The $\fpt$-algorithm is based on the idea of {\em color-coding} by Alon et al.~\cite{Alon1995}. 
The size of a minimal witness for $\G' \geq k$ is bounded by a function of $k$. 
Since those vertices of the same color do not need to induce an independent set, a random coloring will identify a {\em colorful} minimal witness with a good probability. 
}
\appendixproof{\autoref{thm:weakFPT}}
{
\begin{proof}
Let $G$ be the input graph.  
We use the randomized color-coding technique of Alon et al.~\cite{Alon1995}.  
Let us first uniformly randomly color the vertices of $G$ with integers between~$1$ and~$k$, and denote by $\text{col}$ the function giving the color of a vertex according to this random coloring.  
Then, we apply a pruning step, removing all vertices which violate the property of a weak Grundy coloring. 
That is, we remove each vertex $v$ such that $\text{col}(v)=c$ if $\exists c' < c, \neg \exists u \in N(v)$, $\text{col}(u)=c'$.
Equivalently, we keep only the vertices $v$ such that $\forall c \in [\text{col}(v)-1]$, $\exists u \in N(v)$, $\text{col}(u) = c$. 
Note that is well possible that a vertex satisfying the condition at first, no longer satisfies it at a later point, after some of its neighbors are removed.  
Therefore, we apply the pruning until all the vertices satisfy the condition.  
If there is still a vertex colored with~$k$ after this pruning step, then, by construction, there is a weak Grundy coloring achieving color~$k$ in $G$ (by coloring first the vertices $v$ such that $\text{col}(v)=1$, then the vertices $v$ such that $\text{col}(v)=2$, and so on, up to $k$).  

If there is no weak Grundy (minimal) $k$-witness, this computation always rejects. 
Otherwise, it accepts only if a witness is well-colored by the random coloring.  
By Observation~\ref{obs:basic-props2}, a weak Grundy $k$-witness (as a Grundy $k$-witness) has size at most $2^{k-1}$.  
At worst, there is a unique weak Grundy witness of size $2^{k-1}$ admitting a unique coloring.  
The probability to find this witness in one trial is $\frac{1}{k^{2^{k-1}}}$.  
Therefore, by repeating the previous step $\log(\frac{1}{\varepsilon})k^{2^{k-1}}$ times, we find a solution with probability at least $1-\varepsilon$, for any $\varepsilon>0$.  
Overall, the running time is $O(k^{2^{k-1}}(n+m)n)=O^*(2^{2^{O(k)}})$.
\end{proof}
}

We observe that the algorithm of Theorem~\ref{thm:weakFPT} can be derandomized using so-called \emph{universal coloring families}~\cite{Alon1995}.

Unfortunately, the approach used to prove Theorem~\ref{thm:weakFPT} does not work for \PBG because we have no guarantee that the color classes are independent sets.

\subsection{\PBG parameterized by $k$ is in $\fpt$ on special graph classes}

For each fixed~$k$, \PBG can be solved in polynomial time~\cite{Zaker} and thus \PBG parameterized by the number $k$ of colors is in $\xp$. However (unlike \wGC, as seen in Theorem~\ref{thm:weakFPT}), it is unknown whether \PBG is in $\fpt$ when parameterized by $k$. We will next show that it is indeed the case when restricting the instances to $H$-minor-free, chordal and claw-free graphs.\confversion{ Note that \PBG{} is $\np$-complete on chordal graphs~\cite{thesis} and on claw-free graphs~\cite{Zaker-cobip}.}
%

\begin{theorem}\label{prop:FPT-H-minor-free}
 \PBG parameterized by the number of colors is in $\fpt$ for the class of graphs excluding a fixed graph $H$ as a minor.
\end{theorem}
{
\begin{proof}
By Observation~\ref{obs:minimal}, $G$ contains a minimal $k$-witness $H$ as an induced subgraph if and only if $\G(G)\geq k$. 
By Observation~\ref{obs:basic-props2}, a minimal $k$-witness has at most $2^{k-1}$ vertices.
So, the number of minimal $k$-witnesses (up to isomorphism) is bounded by a function of $k$.
Besides, \textsc{$H$-Induced Subgraph Isomorphism} is in $\fpt$ when parameterized by $|V(H)|$ on graphs excluding $H$ as a minor~\cite{FG01}. 
Therefore, one can check if $\G(G)\geq k$ by solving \textsc{$H$-Induced Subgraph Isomorphism} for all minimal $k$-witnesses $H$.
\end{proof}
}

\fullversion{We have the following corollary of the algorithm of Telle and Proskurowski~\cite{TP97}. Note that \PBG{} is $\np$-complete on chordal graphs~\cite{thesis}.}

\begin{theorem}\label{th:chordalFPT}
Let $\mathcal C$ be a graph class for which every member $G$ satisfies
$tw(G)\leq f(\G(G))$ for some function $f$. Then, \PBG{} parameterized by the number of colors is in $\fpt$ on $\mathcal C$. In particular, \PBG{} is in $\fpt$ on chordal graphs.
\end{theorem}
{
\begin{proof}
Since \PBG{} is in $\fpt$ for parameter combination of the number of colors and the treewidth~\cite{TP97}, the
first claim is immediate. Moreover $\omega(G)\leq\G(G)$, hence if
$tw(G)\leq f(\omega(G))$ we have $tw(G)\leq f(\G(G))$. For
any chordal graph $G$, $tw(G)=\omega(G)-1$~\cite{Bodlaender1993}.
\end{proof}
}

The following shows that, unlike the classical \textsc{Coloring} problem, which remains $\np$-hard on degree $4$ graphs, \PBG is $\fpt$ when parameterized by the maximum degree $\Delta(G)$.

\begin{proposition}\label{prop:FPT-delta}
\PBG{} can be solved in time $O\left(nk^{\Delta^{k+1}}\right)=n\Delta^{\Delta^{O(\Delta)}}$ for graphs of maximum degree $\Delta$.
\end{proposition}
{
\begin{proof}
\sloppy
By Observation~\ref{obs:basic-props}, one can enumerate every distance-$k$ neighborhood of each vertex, test every $k$-coloring of this neighborhood, and check if it is a valid Grundy $k$-coloring. 
Every such neighborhood has size at most $\Delta^{k+1}\leq \Delta^{\Delta+2}$ since by Observation~\ref{obs:deg}, $k\leq \Delta+1$. 
Finally, there are at most $k^a$ $k$-colorings of a set of $a$ elements.
\end{proof}
}

\fullversion{We have the following corollary of Proposition~\ref{prop:FPT-delta}. Note that \PBG{} is $\np$-complete on claw-free graphs~\cite{Zaker-cobip}.}

\begin{corollary}\label{cor:fpt-chordal}
Let $\mathcal C$ be a graph class for which every member $G$ satisfies
$\Delta(G)\leqslant f(\G(G))$ for some function $f$. Then,
\PBG{} parameterized by the number of colors is in $\fpt$ for graphs in
$\mathcal C$. In particular, this holds for the class of claw-free
graphs.
\end{corollary}
{
\begin{proof} The first part directly follows from Proposition~\ref{prop:FPT-delta}. For the second part, consider a claw-free graph $G$ and a vertex $v$ of degree $\Delta(G)$ in $G$. Since $G$ is claw-free, the subgraph induced by the neighbors of $v$ has independence number at most~$2$, and hence $\G(G)\geqslant \chi(G)\geqslant\chi(N(v))\geqslant\frac{\Delta(G)}{2}$.
\end{proof}
}

\section{Negative results}\label{sec:negative}

In this section, we present our algorithmic lower bounds.

\subsection{A lower bound for \wGC and \PBG under the ETH}

We now present two similar reductions that (under the ETH) rule out algorithms for \wGC and \PBG with a running time that is sub-double-exponential in $k$ and sub-exponential in the instance size. 
In particular, this shows that the $\fpt$ algorithm for \wGC of Theorem~\ref{thm:weakFPT} has a near-optimal running time, assuming the ETH.

The property "$k \leqslant 1 + w \log n$" (which also holds for weak Grundy colorings \cite{TP97}), entails that a running time $O^*(2^{2^{o(\frac{k}{w})}})$ is in fact subexponential-time $2^{o(n)}$.
Therefore, if a subexponential-time algorithm (in the number of vertices) is proven unlikely, we would immediately obtain the conditional lower bound of $O^*(2^{2^{o(\frac{k}{w})}})$.
Though, it is unclear whether the reductions from the literature on Grundy colorings allow to rule out a subexponential-time algorithm for \PBG (or \wGC) under ETH.
More importantly, what we prove next in Theorem~\ref{thm:weakFPT-lowerbound} is a stronger lower bound, since the treewidth disappears in the denominator of the second exponent.

\begin{theorem}\label{thm:weakFPT-lowerbound}
If \wGC or \PBG is solvable in time $O^*(2^{2^{o(k)}}2^{o(n+m)})$ on graphs with $n$ vertices and $m$ edges, then the ETH fails.
\end{theorem}

{
\begin{proof}
We first give the reduction for \wGC.

In \textsc{Monotone 3-NAE-SAT}, being given a $3$-CNF formula without negation, one is asked to find a truth assignment such that every clause contains a true literal and a false literal.
The \textsc{Monotone 3-NAE-SAT} problem (also called \textsc{Positive 3-NAE-SAT}) with $n$ variables and $m$ clauses, is not solvable in time $2^{o(n+m)}$, unless the ETH fails \cite{JansenLL13}.
More precisely, in the technical report version of the aforementioned paper of Jansen et al.~\cite{Jansen2013}, the authors present a reduction from \textsc{3-SAT} to \textsc{Monotone 3-NAE-SAT}, producing instances with $O(n+m)$ variables and clauses.
We now build from an instance of  \textsc{Monotone 3-NAE-SAT} $\mathcal C=\{C_1, \ldots,C_m\}$ over the variables $X=\{x_1, \ldots, x_n\}$, an equivalent instance of \wGC with $O(n+m)$ vertices and clauses, and $k:=\lceil \log m \rceil + 5$.

We remove, from a binomial tree $T_k$, rooted at $r$, $m$ dominant subtrees $T_3$.
This is possible since the number of such subtrees is $2^{k-3-2}=2^{\lceil \log m \rceil + 5 - 3 - 2}=2^{\lceil \log m \rceil} \geqslant m$.
We call $T$ the tree that we obtain by this process.
We denote by $f_1, \ldots, f_m$ the parents of those removed subtrees, and we link, for each $j \in [m]$, $f_i$ to a new vertex $v(C_j)$ representing the clause $C_j$.
For each $i \in [n]$, we add a star $K_{1,n}$, whose center is denoted by $c$ and whose leaves are denoted by $v(x_i)$, and that represents the variables.
We link each vertex $v(x_i)$ to vertex $v(C_j)$ if variable $x_i$ appears in clause $C_j$. This ends the construction of the graph $G$.

Let us first show that $\G'(G)=k$ if and only if $r$ can be colored $k$.
By Observation~\ref{obs:deg}, the only vertices that can (potentially) be colored with color $k$ are $r$, $r(k-1)$, $c$, and the $v(x_i)$'s.
We already remarked that if $r(k-1)$ can be colored $k$, then, so does $r$ (Lemma~\ref{lem:tree-coloring}).
What remains to prove is that neither $c$ nor any of the $v(x_i)$'s cannot be colored $k$.
The neighbors of a vertex $v(x_i)$ are $c$ and some vertices $v(C_j)$, whose degree is bounded by $4$ (recall that the clauses contain at most three variables).
Thus, $v(x_i)$ can have in its neighborhood at most six distinct colors, and its color can be at most $7$.
Similarly, the neighbors of $c$ are the $v(x_i)$'s, so the color of vertex $c$ can be at most $8$.
We can assume that $\lceil \log m \rceil > 3$ (and, $k>8$) since otherwise the instance is of constant size.
Therefore, $\G'(G)=k$ if and only if $r$ can be colored $k$, which means that, by applying Lemma~\ref{lem:tree-coloring-partial} with induced subtree $T$, we have $\G'(G)=k$ if and only if $v(C_j)$ can be colored $3$, for each $j \in [m]$, without first coloring any of the $f_j$'s.

Now, suppose that $\mathcal C$ is satisfiable.
Let $\psi$ be a satisfying truth assignment of $\mathcal C$.
Then, we can color each vertex $v(C_j)$ with color $3$ in the following way.
We first color $c$ with color $1$.
Then, for each $i \in [n]$, we color $v(x_i)$ with $1$ if $x_i$ is set to false by $\psi$, and with $2$ if it is set to true.
Recall that the weak Grundy coloring does not need to be proper.
As each clause $C_j$ has at least one variable $x_{i_1}$ set to true and at least one variable $x_{i_2}$ set to false, $v(C_j)$ has in its neighborhood a vertex $v(x_{i_1})$ colored $2$ and a vertex $v(x_{i_2})$ colored $1$.
Hence, $v(C_j)$ can be colored $3$; moreover, we have not colored any vertex $f_j$, and we are done.

Conversely, suppose that $v(C_j)$ can be colored $3$, for each $j \in [m]$, without coloring first any of the $f_j$'s.
Then, in the neighborhood of each $v(C_j)$ deprived of the $f_j$'s, there should be one vertex $v(x_{i_1})$ colored $2$ and one vertex $v(x_{i_2})$ colored $1$.
Therefore, the truth assignment $\psi$ setting $x_i$ to true if $v(x_i)$ has been colored $2$ and to false if $v(x_i)$ has been colored $1$ or has not been colored, satisfies $\mathcal C$.

In conclusion, we showed that $\G'(G)=k$ if and only if $\mathcal C$ is satisfiable.
The number $N$ of vertices of the graph $G$ is bounded by $n+1+2^{\lceil \log m \rceil+4} \leqslant n+32m+1=O(n+m)$.
The number of edges $M$ is bounded by $n+3m+2^{\lceil \log m \rceil+4} \leqslant n+35m=O(n+m)$
Thus, solving \wGC in time $O^*(2^{2^{o(k)}}2^{o(N+M)})=O^*(2^{o(m)}2^{o(n+m)})=O^*(2^{o(n+m)})$ would solve \textsc{Monotone 3-NAE-SAT} in subexponential-time, disproving the ETH.

For \PBG, we use a similar reduction by replacing the star $K_{1,n}$ encoding the variables by a matching of $n$ edges $v(\neg x_i)v(x_i)$ where $v(\neg x_i)$ is a new vertex having only one neighbor: $v(x_i)$.
Then, the proof carries over: in a Grundy coloring, one could color $v(x_i)$ with color $1$ or $2$, by first coloring $v(\neg x_i)$ with color $1$.
\end{proof}
}

The behavior shown by \wGC is rare, and up to our knowledge, the only other known example for which an $O^*(2^{2^{O(k)}})$ is optimal under the ETH (with $k$ the natural parameter)
is the \textsc{Edge Clique Cover} problem~\cite{CyganPP13}. For the \textsc{Edge Clique Cover} problem, where one wants to cover all the edges of a graph by a minimum number $k$ of cliques, only an algorithm running in time $O^*(2^{2^{o(k)}}2^{o(n)})$ would disprove ETH.
The number of edges in the produced instance of \textsc{Edge Clique Cover} has to be superlinear.
Indeed, otherwise the maximum clique would be of constant size, and the parameter $k$ would be at least linear in the number of vertices $n$, when it should in fact be logarithmic in $n$.
Therefore, \wGC seems to be the first problem for which an $O^*(2^{2^{o(k)}}2^{o(n+m)})$-algorithm is shown to be unlikely, while an $O^*(2^{2^{O(k)}})$-algorithm exists.

\subsection{Lower bound on the treewidth dependency for \PBG and \wGC}\label{subsec:ETH}
Let us recall that the algorithm for \PBG and \wGC running in time $n^{O(w^2)}$ of Telle and Proskurowski comes from a $2^{O(w k)}n$-algorithm and the fact that $k \leqslant w \log n+1$~\cite{TP97}.

An interesting observation is that an algorithm for \PBG or \wGC running in time $O^*(k^{O(w)})=O^*(2^{O(w\log k)})$, where $w$ is the treewidth of the input graph, would imply an $\fpt$ algorithm for the parameter treewidth alone.

\begin{observation}\label{obs:tw-fpt}
If \PBG or \wGC can be solved in time $O^*(k^{O(w)})$ on instances of treewidth $w$, then it can be solved in time $O^*(2^{O(w \log w)})$.
\end{observation}
\begin{proof}
Since, as mentioned before, $k \leqslant w \log n+1$~\cite{TP97} and using the fact that $\forall x,y > 0, (\log x)^y \leq y^{2y}x$, we have $O^*(k^{O(w)})=O^*(w^{O(w)}(\log n)^{O(w)})=O^*(w^{O(w)})=O^*(2^{O(w \log w)})$.
\end{proof}


Note that there are $k^w$ possible $k$-colorings of a bag of size $w$, hence an algorithm for \PBG or \wGC running in time $O^*(k^{O(w)})$ could be based on dynamic programming over a tree decomposition (and would greatly improve over the running time of the algorithm of~\cite{TP97}). Although we do not know whether such an algorithm exists, we now show that (assuming the ETH), one cannot hope for a significantly better running time (even when replacing the treewidth by the larger parameter ``feedback vertex set number''). The reduction has some similarities with the reduction from Theorem~\ref{thm:weakFPT-lowerbound}, but it is more involved since we need to additionally lower the value of the treewidth.

\begin{theorem}\label{thm:tw-LB-ETH}
If \PBG or \wGC is solvable in time $O^*(2^{o(w \log w)})$ on graphs with feedback vertex set at most~$w$, then the ETH fails.
\end{theorem}

{
\begin{proof}
We describe the proof for \PBG, but the same proof also works for \wGC.

We build from an instance of \Sat an equivalent instance of \PBG with subexponentially many vertices and sublinear feedback vertex set number.
We rely on the grouping technique (similarly to \cite{LokshtanovSODA}) that uses the fact that the number of permutations over a slightly sublinear number of elements is still exponential. 
We also make multiple applications of Lemma~\ref{lem:tree-coloring-partial}. 

Let $\mathcal C = \{C_1,\ldots,C_m\}$ be the $m$ clauses of an instance of \Sat over the set of variables $X=\{x_1,\ldots,x_n\}$.
Let $q$ be a positive integer that we will fix later.
We partition arbitrarily $X$ into $q$ sets $X_1, \ldots, X_q$ called \emph{groups}, each of size at most $\lceil \frac{n}{q} \rceil$.
A \emph{group assignment} is a truth assignment of the variables of $X_h$ for some $h \in [q]$.
A group assignment \emph{satisfies} a clause if it sets to true at least one of its literals (even if some variables of the clause are not instantiated).
By potentially adding dummy variables, we may assume that $|X_h|=\lceil \frac{n}{q} \rceil$, for each $h \in [q]$.
We also fix an arbitrary order of the variables within each group $X_h$, so that an assignment of $X_h$ can be seen as a word of $\{0,1\}^{\lceil \frac{n}{q} \rceil}$. 
Let $t=\lceil 3n/(q \log \frac{n}{q}) \rceil$ and recall that $\mathfrak{S}_t$ is the symmetric group.
We fix an arbitrary one-to-one function $\zeta: \{0,1\}^{\lceil \frac{n}{q} \rceil} \rightarrow \mathfrak{S}_t$ mapping a group assignment to a permutation over $t$ elements.
Such a function exists, since $|\mathfrak{S}_t|=t!>(\frac{t}{3})^t \geqslant 2^{3n(\log \frac{n}{q} - \log \log \frac{n}{q})/(q \log \frac{n}{q})}>2^{\lceil n/q \rceil}$.
Finally, we set $s=\lceil \log m \rceil + 2t + 4$.

We now describe the instance graph $G$ of \PBG.
We remove, from a binomial tree $T_s$ rooted at $r$, $m$ (arbitrary) dominant subtrees $T_{t+2}$.
This is possible since, in $T_s$, there are $2^{s-t-4}=2^{\lceil \log m \rceil + t} \geqslant m$ dominant trees $T_{t+2}$.
We denote by $f_1,\ldots,f_m$ the $m$ parents of those $m$ removed subtrees.
We call $T$ the tree that we have obtained so far.
For each clause $C_j$ ($j \in [m]$) and for each group assignment $\tau$ (of some group $X_h$) satisfying $C_j$, we add a vertex $v(j,\tau)$ that we link to $f_j$.
We denote by $I_j$ the set of vertices $v(j,\cdot)$.
Vertex $v(j,\tau)$ also becomes the \emph{root} of a binomial tree $T_{t+2}$ from which we remove the dominant subtree of each of its children (except for the child $v(j,\tau)(1)$ which is a leaf and therefore has no dominant subtree).
We call that tree $T(j,\tau)$.
Now, for each group $X_h$ ($h \in [q]$), we add a clique $S_h=\{s_h^1,\ldots,s_h^t\}$ on $t$ vertices.
For each vertex $v(j,\tau)$, if $\tau$ is an assignment of the group $X_h$ (for some $h \in [q]$) and $\sigma=\zeta(\tau)$, we link $v(j,\tau)(p+1)$ to $s_h^{\sigma(p)}$, for each $p \in [t]$.

This ends the construction of graph $G$ (see Figure~\ref{fig:hardness}\confversion{ in the appendix} for an illustration).
The number $N$ of vertices of $G$ is upper-bounded by $mq2^{\lceil \frac{n}{q} \rceil}2^{t+1}+2^{s-1}+qt=O(mq2^{2t+\frac{n}{q}})$. The set $\bigcup_{h \in [q]}S_h$ is a feedback vertex set of $G$ of size $qt$.

{
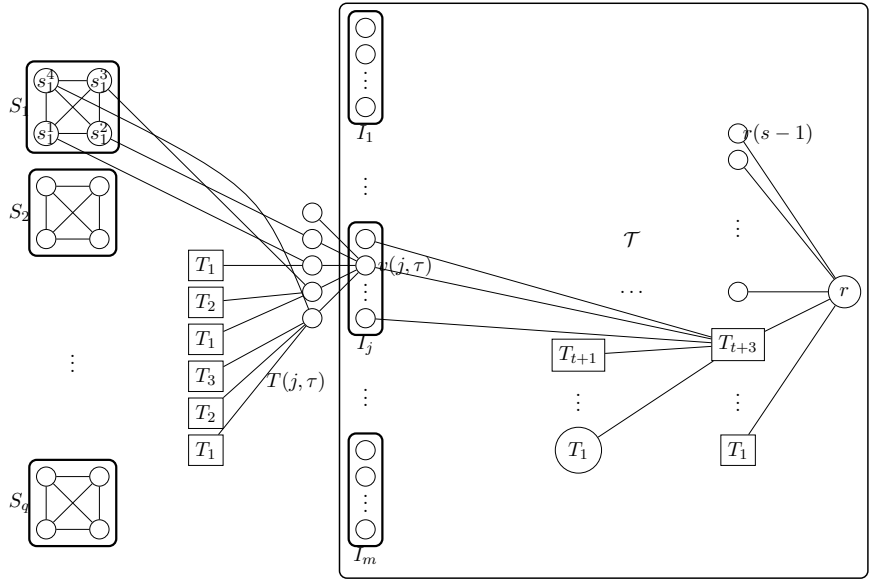
\begin{figure}
\centering
\scalebox{0.7}{\begin{tikzpicture}


\begin{scope}
\node (ta) at (-0.5,0.5) {$S_1$};
\node[draw,circle,inner sep=-0.03cm] (a1) at (0,0) {$s_1^1$};
\node[draw,circle,inner sep=-0.03cm] (a2) at (1,0) {$s_1^2$};
\node[draw,circle,inner sep=-0.03cm] (a3) at (0,1) {$s_1^4$};
\node[draw,circle,inner sep=-0.03cm] (a4) at (1,1) {$s_1^3$};
\node[fit=(a1) (a3) (a2) (a4), draw,  very thick, rectangle, rounded corners] {};
\draw (a1) -- (a2) -- (a4) -- (a3) -- (a1) -- (a4) ;
\draw (a2) -- (a3) ;
\end{scope}

\begin{scope}[yshift=-2cm]
\node (tb) at (-0.5,0.5) {$S_2$};
\node[draw,circle] (b1) at (0,0) {};
\node[draw,circle] (b2) at (1,0) {};
\node[draw,circle] (b3) at (0,1) {};
\node[draw,circle] (b4) at (1,1) {};
\node[fit=(b1) (b3) (b2) (b4), draw,  very thick, rectangle, rounded corners] {};
\draw (b1) -- (b2) -- (b4) -- (b3) -- (b1) -- (b4) ;
\draw (b2) -- (b3) ;
\end{scope}

\node () at (0.5,-4.25) {$\vdots$};

\begin{scope}[yshift=-7.5cm]
\node (tc) at (-0.5,0.5) {$S_q$};
\node[draw,circle] (c1) at (0,0) {};
\node[draw,circle] (c2) at (1,0) {};
\node[draw,circle] (c3) at (0,1) {};
\node[draw,circle] (c4) at (1,1) {};
\node[fit=(c1) (c3) (c2) (c4), draw,  very thick, rectangle, rounded corners] {};
\draw (c1) -- (c2) -- (c4) -- (c3) -- (c1) -- (c4) ;
\draw (c2) -- (c3) ;
\end{scope}


\begin{scope}[xshift=6cm, yshift=1cm]
\node[draw,circle] (i11) at (0,1) {};
\node[draw,circle] (i12) at (0,0.5) {};
\node () at (0,0.1) {$\vdots$};
\node[draw,circle] (i1l) at (0,-0.5) {};
\node[fit=(i11) (i12) (i1l), draw,  very thick, rectangle, rounded corners] (i1) {};
\node () at (0,-1) {$I_1$};
\end{scope}

\node () at (6,-0.9) {$\vdots$};

\begin{scope}[xshift=6cm,yshift=-3cm]
\node[draw,circle] (i21) at (0,1) {};
\node[draw,circle] (i22) at (0,0.5) {};
\node (ti22) at (0.75,0.5) {$v(j,\tau)$};
\node () at (0,0.1) {$\vdots$};
\node[draw,circle] (i2l) at (0,-0.5) {};
\node[fit=(i21) (i22) (i2l), draw, very thick, rectangle, rounded corners] (ij) {};
\node () at (0,-1) {$I_j$};
\end{scope}

\node () at (6,-4.9) {$\vdots$};

\begin{scope}[xshift=6cm,yshift=-7cm]
\node[draw,circle] (im1) at (0,1) {};
\node[draw,circle] (im2) at (0,0.5) {};
\node () at (0,0.1) {$\vdots$};
\node[draw,circle] (iml) at (0,-0.5) {};
\node[fit=(im1) (im2) (iml), draw, very thick, rectangle, rounded corners] (im) {};
\node (tim) at (0,-1) {$I_m$};
\end{scope}


\begin{scope}[xshift=2cm,yshift=-0.5cm]
\node[draw,circle] (s1) at (3,-1) {};
\node[draw,circle] (s2) at (3,-1.5) {};
\node[draw,circle] (s3) at (3,-2) {};
\node[draw,circle] (s4) at (3,-2.5) {};
\node[draw,circle] (s5) at (3,-3) {};
\end{scope}

\begin{scope}[yshift=-0.7cm]
\node[draw,rectangle] (s31) at (3,-1.8) {$T_1$};

\node[draw,rectangle] (s41) at (3,-2.5) {$T_2$};
\node[draw,rectangle] (s42) at (3,-3.2) {$T_1$};

\node[draw,rectangle] (s51) at (3,-3.9) {$T_3$};
\node[draw,rectangle] (s52) at (3,-4.6) {$T_2$};
\node[draw,rectangle] (s53) at (3,-5.3) {$T_1$};

\node (tjtau) at (4.7,-4) {$T(j,\tau)$};
\end{scope}

\draw (i22) -- (s1) ;
\draw (i22) -- (s2) ;
\draw (i22) -- (s3) ;
\draw (i22) -- (s4) ;  
\draw (i22) -- (s5) ; 

\draw (s2) -- (a2);
\draw (s3) -- (a1);
\draw (s4) -- (a4);
\draw[-] (s5) .. controls (4,-1) .. (a3);

\draw (s3) -- (s31) ;

\draw (s4) -- (s41) ;
\draw (s4) -- (s42) ;

\draw (s5) -- (s51) ;
\draw (s5) -- (s52) ;
\draw (s5) -- (s53) ; 

\begin{scope}[xshift=14cm,yshift=-4cm]
\node[draw,circle] (t) at (1,1) {$r$};

\node at (-3,1) {$\ldots$};

\node[draw,circle] (t1) at (-1,4) {};
\node (rs) at (-0.25,4) {$r(s-1)$};
\node[draw,circle] (t2) at (-1,3.5) {};
\node () at (-1,2.3) {$\vdots$};
\node[draw,circle] (tl) at (-1,1) {};
\node[draw,rectangle] (tt3) at (-1,0) {$T_{t+3}$};
\node[] (tt2) at (-1,-1) {$\vdots$};
\node[draw,rectangle] (tt1) at (-1,-2) {$T_1$};

\begin{scope}[yshift=-6cm]
\node[draw,circle] (t11) at (-4,4) {$T_1$};
\node[draw,rectangle] (tt13) at (-4,5.8) {$T_{t+1}$};
\node (tt12) at (-4,5) {$\vdots$};
\end{scope}


\end{scope}

\begin{scope}[xshift=1cm,yshift=-2cm]
\node[fit=(i1) (tim) (t), draw, thick, rectangle, rounded corners] () {};
\node () at (10,0) {$\mathcal T$} ;
\end{scope}

\draw (t) -- (t1) ;
\draw (t) -- (t2) ;
\draw (t) -- (tl) ;
\draw (t) -- (tt3) ;
\draw (t) -- (tt1) ;


\draw(tt3) -- (i21) ;
\draw(tt3) -- (i22) ;
\draw(tt3) -- (i2l) ;

\draw(tt3) -- (t11);

\draw(tt3) -- (tt13);


\end{tikzpicture}}

\caption{A sample of the construction of graph $G$. The edges incident to the rectangular boxes containing $T_i$s are only incident to the root of the tree. For the sake of readability, only one $T(j,\tau)$ is represented. Here, $v(j,\tau)$ represents an assignment of group $X_1$ mapped to the permutation $\sigma=(12)$.}
\label{fig:hardness}
\end{figure}
}

We now show that $\mathcal C$ is satisfiable if and only if $\Gamma(G) \geqslant s$. 
The proof goes as follows:\\
(1) $\Gamma(G) \geqslant s$ if and only if $r$, the root of $T$, can be colored with color $s$;\\ 
(2) by Lemma~\ref{lem:tree-coloring-partial} on $T$, this is equivalent to color a vertex in each $I_j$ with color $t+2$;\\
(3) by Lemma~\ref{lem:tree-coloring-partial} applied to the set of all trees $T(j,\tau)$, this is equivalent to a property $(\mathcal P)$ (that we will define later) on the coloring of the cliques $S_h$'s;\\
(4) $\mathcal C$ is satisfiable implies $(\mathcal P)$;\\
(5) $(\mathcal P)$ implies $\mathcal C$ is satisfiable.

First, we show the equivalence (1) that $\Gamma(G) \geqslant s$ if and only if $r$ can be colored $s$. Assume that $\Gamma(G) \geqslant s$ (the other implication is trivial). By Observation~\ref{obs:deg}, the only vertices (besides $r$) whose degree are (or at least may be) sufficient to be colored $s$ are $r(s-1)$ (but we already noticed in Lemma~\ref{lem:tree-coloring} that $r(s-1)$ can be colored $s$ if and only if this is the case for $r$) and the vertices of the cliques $S_h$'s.
The vertices in $N(S_h)$ have degree at most $t+1$, hence their color can be at most $t+2$.
Thus, the number of distinct colors that a vertex of $S_h$ can see in its neighborhood is at most $t+2+(t-1)=2t-1$. Hence, its color cannot exceed $2t$, which is strictly smaller than $s$. Hence, $r$ (or $r(s-1)$) has to be the vertex colored $s$.

To see that~(2)holds, observe that, by Lemma~\ref{lem:tree-coloring-partial} applied to the induced tree $T$ and the set of parents $F=\{f_1,\ldots,f_m\}$ of removed subtrees, $\Gamma(G) \geqslant s$ if and only if there is a Grundy coloring of $G-T$ coloring at least one vertex of $I_j$ with color $t+2$, without coloring any vertex of any $I_j$ with color $t+3$.
This latter condition can be omitted, since, in order to color a vertex with color $t+2$, first coloring other vertices with color $t+3$ or more is not helpful.
We can now remove $T$ from the graph $G$ and equivalently ask if one can color with $t+2$ at least one vertex in each set $I_j$ ($j \in [m]$) in this new graph $G'$.

For each $j \in [m]$, for every vertex $v(j,\tau) \in I_j$, we apply Lemma~\ref{lem:tree-coloring-partial} with the induced tree $T(j,\tau)$ and the set of parents $\{v(j,\tau)(2),\ldots,v(j,\tau)(t+1)\}$: $v(j,\tau)$ can be colored with color $t+2$ if and only if $s_h^{\sigma(p)}$ can be colored with $p$, for each $p \in [t]$, without coloring first any vertex of $T(j,\tau)$ (where $\sigma=\zeta(\tau)$ and $\tau$ is an assignment of the group $X_h$).
As $S_h$ is a $t$-clique, receiving each color from $1$ to $t$ cannot benefit from coloring vertices of $N(S_h)$ first.
Thus, we will assume that all the $S_h$'s are colored first.

We call $(\mathcal P)$ the property:\\
$\forall j \in [m]$, $\exists v(j,\tau) \in I_j$, such that $\forall p \in [t]$, $s_h^{\zeta(\tau)(p)}$ has color $p$.

So far, we have shown that $\G(G)\geqslant s$ if and only if $(\mathcal P)$ holds.
We now show that $(\mathcal P)$ holds if and only if $\mathcal C$ is satisfiable.

Assume $\mathcal C$ is satisfiable.
Let $\psi$ be a satisfying global assignment.
Let $\tau_h$ be the projection of $\psi$ to $X_h$ for each $h \in [q]$, and $\sigma_h=\zeta(\tau_h)$.
We color the cliques $S_h$'s such that $s_h^{\sigma_h(p)}$ is colored $p$, for each $p \in [t]$; that is, we first color $s_h^{\sigma_h(1)}$, then $s_h^{\sigma_h(2)}$, and so on, up to $s_h^{\sigma_h(t)}$.
Now, for each clause $C_j$, there is a literal of $C_j$ which is set to true by $\psi$. 
Say, this literal is on a variable of $X_h$ for some $h \in [q]$.
Then, the group assignment $\tau_h$ satisfies $C_j$.
Therefore, for each $j \in [m]$, vertex $v(j,\tau_h) \in I_j$ exists and $\forall p \in [t]$, $s_h^{\zeta(\tau_h)(p)}$ has color $p$. 

Assume now that the $S_h$'s has been colored first, and such that $(\mathcal P)$ holds.
For each $h \in [q]$, let $\sigma_h$ be the permutation of $\mathfrak{S}_t$ such that $\sigma_h(p)$ is defined as the index in $S_h$ of the vertex colored $p$.
This is well-defined since the vertices of the clique $S_h$ have each color of $[t]$ exactly once.
Now, let $\tau_h$ be the assignment of the group $X_h$ such that $\sigma_h=\zeta(\tau_h)$.
Group assignment $\tau_h$ is unique since $\zeta$ is one-to-one, and exists for $(\mathcal P)$ to hold.
Let $\psi$ be the global assignment whose projection to each $X_h$ is $\tau_h$.
By $(\mathcal P)$, for each $j \in [m]$, there is vertex $v(j,\tau) \in I_j$ such that $\forall p \in [t]$, $s_h^{\zeta(\tau)(p)}$ has color $p$, for some $h \in [q]$.
This vertex has in fact to be $v(j,\tau_h)$ since clique $S_h$ has been colored such that $s_h^{\sigma_h(p)}$ has color $p$.
So, the group assignment $\tau_h$ satisfies $C_j$.
Therefore, $\psi$ is a satisfying global assignment, hence $\mathcal C$ is satisfiable.

\medskip

Suppose there is an algorithm solving \PBG on graphs with $N$ vertices and feedback vertex set $w$ in time $2^{o(w \log w)}N^c$ for some constant $c$. Recall that in $G$, we have $N=O(mq2^{2t+\frac{n}{q}})$ and $w\leq qt$.
Assuming the ETH, there is a constant $s_3 > 0$ such that \Sat (even \textsc{3-SAT}) is not solvable in time $O(2^{s_3 n})$.
Setting $q=\lceil \frac{2c}{s_3} \rceil$, one can solve \Sat in time $O(2^{qt \log (qt)}(mq)^c2^{2t}2^{\frac{cn}{q}})=O(2^{o(3n(\frac{\log n - \log \log n + \log \log q + \log 3}{\log n - \log q}))}(mq)^c2^{\frac{6n}{q \log n/q}}2^{\frac{s_3 n}{2}})=O((mq)^c2^{o(n)}2^{o(n)}2^{\frac{s_3 n}{2}})$ that is $O(2^{s_3 n})$, contradicting the ETH.
\end{proof}

Note that in reduction of the proof of Theorem~\ref{thm:tw-LB-ETH}, we had $2^{o(s \log s)}=2^{o(n+m)}$, so we even proved that \PBG and \wGC cannot be solved in time $O^*(2^{o(k \log k)}2^{o(w \log w)})$ unless the ETH fails (where $k$ is the number of colors).

\subsection{\cGC is NP-hard for $k=7$ colors}

Minimal connected Grundy $k$-witnesses, contrary to minimal Grundy $k$-witnesses (Observation~\ref{obs:basic-props2}), have arbitrarily large order: for instance, the cycle $C_n$ of order~$n$ ($n > 4$, $n$ odd) has a Grundy $3$-witness of order~$4$, but its unique \emph{connected} Grundy $3$-witness is of order~$n$: the whole cycle.  

Observe that $\cG(G) \leqslant 2$ if and only if $G$ is bipartite. 
Hence, \cGC is polynomial-time solvable for any $k \leqslant 3$. 
However, we will now show that the problem is already $\np$-hard when $k=7$, contrary to \PBG and \wGC which are polynomial-time solvable whenever $k$ is a constant (Corollary~\ref{cor:gr-xp} and Theorem~\ref{thm:weakFPT}). 
Thus, in the terminology of parameterized complexity, \cGC is $\paranp$-hard. 
 


\begin{theorem}\label{thm:connected-not-xp}
\cGC is $\np$-hard even for $k=7$.
\end{theorem}
\begin{proof}
We give a reduction from \textsc{3-SAT 3-OCC}, an $\np$-complete restriction of \textsc{3-SAT} where each variable appears in at most three clauses~\cite{Papadimitriou1994}, to \cGC with $k=7$.
We first give the intuition of the reduction. 
The construction consists of a tree-like graph of constant order (resembling binomial tree $T_6$) whose root is adjacent to two vertices of a $K_6$ (this constitutes $W$) and contains three special vertices $a_4$, $a_{21}$, and $a_{24}$ (which will have to be colored with colors~$1$, $3$, and~$2$ respectively), a connected graph $P_1$ which encodes the variables and a path $P_2$ which encodes the clauses. 
One in every three vertices of $P_2$ is adjacent to $a_4$, $a_{21}$ and $a_{24}$. 
To achieve color~$7$, we will need to color those vertices with color strictly greater than~$3$. 
This will be possible if and only if the assignment corresponding to the coloring of $P_1$ satisfies all the clauses.  

We now formally describe the construction.
Let $\phi=(X=\{x_1,\ldots,x_n\},\mathcal C=\{C_1,\ldots,C_m\})$ be an instance of \textsc{3-SAT 3-OCC} where no variable appears always as the same literal.
$P_1=(\{i_1,i_2,v\}\cup\{v_i,\overline{v_i} $ $|$ $ i \in [n]\},\{\{i_1,i_2\},\{i_2,v\}\} \cup \{\{v,v_i\}\cup \{v,\overline{v_i}\}\cup \{v_i,\overline{v_i}\}$ $|$ $ i \in [n]\})$ consists of $n$ triangles sharing the vertex $v$.
$P_2=(\{p_j $ $|$ $j \in [3m-1]\},\{\{p_j,p_{j+1}\} $ $|$ $ j \in [3m-2])$ consists of a path of length $3m-1$.
For each $j \in [m]$ and $i \in [n]$, $c_j\stackrel{def}{=}p_{3j-1}$ is adjacent to $v_i$ if $x_i$ appears positively in $C_j$, and is adjacent to $\overline{v_i}$ if $x_i$ appears negatively in $C_j$.
For each $j \in [m]$, $c_j$ is adjacent to $a_4$, $a_{21}$, and $a_{24}$.   

\begin{figure}[ht]
\centering
\scalebox{0.8}{\begin{tikzpicture}
\node[draw,circle,double,inner sep = -0.28cm] (b1) at (0,2) {$a_4$} ;
\node[draw,circle,inner sep = -0.28cm] (i1) at (2,2.33) {$i_1$} ;
\node[draw,circle,inner sep = -0.28cm] (i2) at (4,2.66) {$i_2$} ;
\node[draw,circle,inner sep = -0.22cm] (v) at (6,3) {$v$} ;

\node[draw,circle,inner sep = -0.28cm] (v1) at (2.6,2) {$v_1$} ;
\node[draw,circle,inner sep = -0.31cm] (nv1) at (3.4,2) {$\overline{v_1}$} ;

\node[draw,circle,inner sep = -0.28cm] (v2) at (4.6,2) {$v_2$} ;
\node[draw,circle,inner sep = -0.31cm] (nv2) at (5.4,2) {$\overline{v_2}$} ;

\node[draw,circle,inner sep = -0.28cm] (v3) at (6.6,2) {$v_3$} ;
\node[draw,circle,inner sep = -0.31cm] (nv3) at (7.4,2) {$\overline{v_3}$} ;

\node[draw,circle,inner sep = -0.28cm] (v4) at (8.6,2) {$v_4$} ;
\node[draw,circle,inner sep = -0.31cm] (nv4) at (9.4,2) {$\overline{v_4}$} ;

\node[draw,circle,inner sep = -0.28cm] (b2) at (12,2) {$a_6$} ;

\draw (b1) -- (i1) -- (i2) -- (v) -- (b2) ;
\draw (v1) -- (nv1) -- (v) -- (v1) ; 
\draw (v2) -- (nv2) -- (v) -- (v2) ; 
\draw (v3) -- (nv3) -- (v) -- (v3) ; 
\draw (v4) -- (nv4) -- (v) -- (v4) ; 

\begin{scope}[yshift=-0.2cm]
\node[draw,circle,inner sep = -0.28cm] (c1) at (0,0) {$a_9$} ;
\node[draw,circle] (c2) at (1,0) {} ;
\node[draw,circle,inner sep = -0.28cm] (c3) at (2,0) {$c_1$} ;
\node[draw,circle] (c4) at (3,0) {} ;
\node[draw,circle] (c5) at (4,0) {} ;
\node[draw,circle,inner sep = -0.28cm] (c6) at (5,0) {$c_2$} ;
\node[draw,circle] (c7) at (6,0) {} ;
\node[draw,circle] (c8) at (7,0) {} ;
\node[draw,circle,inner sep = -0.28cm] (c9) at (8,0) {$c_3$} ;
\node[draw,circle] (c10) at (9,0) {} ;
\node[draw,circle] (c11) at (10,0) {} ;
\node[draw,circle,inner sep = -0.28cm] (c12) at (11,0) {$c_4$} ;
\node[draw,circle,inner sep = -0.33cm] (c13) at (12,0) {$a_{11}$} ;
\end{scope}

\draw (c3) -- (b1) -- (c6) ;
\draw (c9) -- (b1) -- (c12) ;

\draw (c1) -- (c2) -- (c3) -- (c4) -- (c5) -- (c6) -- (c7) -- (c8) -- (c9) -- (c10) -- (c11) -- (c12) -- (c13) ;

\draw (v1) -- (c3) ;
\draw (nv2) -- (c3) ;
\draw (v3) -- (c3) ;

\draw (v1) -- (c6) ;
\draw (v2) -- (c6) ;
\draw (nv4) -- (c6) ;

\draw (nv1) -- (c9) ;
\draw (nv3) -- (c9) ;
\draw (v4) -- (c9) ;

\draw (v2) -- (c12) ;
\draw (nv3) -- (c12) ;
\draw (v4) -- (c12) ;

\draw[densely dashed] (v)+(0,-0.7) ellipse (4.7cm and 1cm);
\path (v)+(4.8,-0.2) node {$P_1$};
\draw[densely dashed] (c7) ellipse (5.5cm and 0.7cm);
\path (c7)+(5,0.6) node {$P_2$};
\end{tikzpicture}}
\caption{$P_1$ and $P_2$ for the instance $\{x_1 \lor \neg x_2 \lor x_3\}, \{x_1 \lor x_2 \lor \neg x_4\}, \{\neg x_1 \lor x_3 \lor x_4\}, \{x_2 \lor \neg x_3 \lor x_4\}$.}
\label{fig:connected-variables-clauses}
\end{figure}
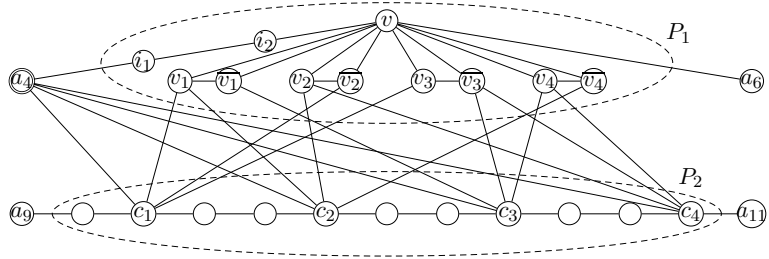

Intuitively, setting a literal to true consists of coloring the corresponding vertices with~$3$.
Therefore, a clause $C_j$ is satisfied if $c_j$ has a $3$ among its neighbors.
To actually satisfy a clause, one has to color $c_j$ with $4$ or higher.
Thus, $c_j$ must also see a $2$ in its neighborhood. 
We will show that the unique way of doing so is to color $p_{3j-2}$ with $2$, so all the clauses have to be checked along the path $P_2$.

{
We give, in Figure~\ref{fig:connected-coloring-variables-clauses}, a coloring of $P_1$ corresponding to a truth assignment of the instance SAT formula.
One can check that when going along $P_2$ all the $c_j$'s are colored with color~$4$.

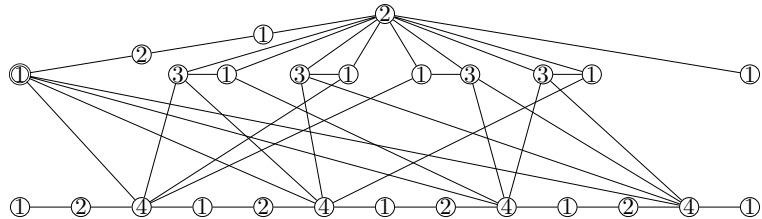
\begin{figure}[ht]
\centering
\scalebox{0.8}{\begin{tikzpicture}
\node[draw,circle,double,inner sep = -0.22cm] (b1) at (0,2) {$1$} ;
\node[draw,circle,inner sep = -0.22cm] (i1) at (2,2.33) {$2$} ;
\node[draw,circle,inner sep = -0.22cm] (i2) at (4,2.66) {$1$} ;
\node[draw,circle,inner sep = -0.22cm] (v) at (6,3) {$2$} ;

\node[draw,circle,inner sep = -0.22cm] (v1) at (2.6,2) {$3$} ;
\node[draw,circle,inner sep = -0.22cm] (nv1) at (3.4,2) {$1$} ;

\node[draw,circle,inner sep = -0.22cm] (v2) at (4.6,2) {$3$} ;
\node[draw,circle,inner sep = -0.22cm] (nv2) at (5.4,2) {$1$} ;

\node[draw,circle,inner sep = -0.22cm] (v3) at (6.6,2) {$1$} ;
\node[draw,circle,inner sep = -0.22cm] (nv3) at (7.4,2) {$3$} ;

\node[draw,circle,inner sep = -0.22cm] (v4) at (8.6,2) {$3$} ;
\node[draw,circle,inner sep = -0.22cm] (nv4) at (9.4,2) {$1$} ;

\node[draw,circle,inner sep = -0.22cm] (b2) at (12,2) {$1$} ;

\draw (b1) -- (i1) -- (i2) -- (v) -- (b2) ;
\draw (v1) -- (nv1) -- (v) -- (v1) ; 
\draw (v2) -- (nv2) -- (v) -- (v2) ; 
\draw (v3) -- (nv3) -- (v) -- (v3) ; 
\draw (v4) -- (nv4) -- (v) -- (v4) ; 

\begin{scope}[yshift=-0.2cm]
\node[draw,circle,inner sep = -0.22cm] (c1) at (0,0) {$1$} ;
\node[draw,circle,inner sep = -0.22cm] (c2) at (1,0) {$2$} ;
\node[draw,circle,inner sep = -0.22cm] (c3) at (2,0) {$4$} ;
\node[draw,circle,inner sep = -0.22cm] (c4) at (3,0) {$1$} ;
\node[draw,circle,inner sep = -0.22cm] (c5) at (4,0) {$2$} ;
\node[draw,circle,inner sep = -0.22cm] (c6) at (5,0) {$4$} ;
\node[draw,circle,inner sep = -0.22cm] (c7) at (6,0) {$1$} ;
\node[draw,circle,inner sep = -0.22cm] (c8) at (7,0) {$2$} ;
\node[draw,circle,inner sep = -0.22cm] (c9) at (8,0) {$4$} ;
\node[draw,circle,inner sep = -0.22cm] (c10) at (9,0) {$1$} ;
\node[draw,circle,inner sep = -0.22cm] (c11) at (10,0) {$2$} ;
\node[draw,circle,inner sep = -0.22cm] (c12) at (11,0) {$4$} ;
\node[draw,circle,inner sep = -0.22cm] (c13) at (12,0) {$1$} ;
\end{scope}

\draw (c3) -- (b1) -- (c6) ;
\draw (c9) -- (b1) -- (c12) ;

\draw (c1) -- (c2) -- (c3) -- (c4) -- (c5) -- (c6) -- (c7) -- (c8) -- (c9) -- (c10) -- (c11) -- (c12) -- (c13) ;

\draw (v1) -- (c3) ;
\draw (nv2) -- (c3) ;
\draw (v3) -- (c3) ;

\draw (v1) -- (c6) ;
\draw (v2) -- (c6) ;
\draw (nv4) -- (c6) ;

\draw (nv1) -- (c9) ;
\draw (nv3) -- (c9) ;
\draw (v4) -- (c9) ;

\draw (v2) -- (c12) ;
\draw (nv3) -- (c12) ;
\draw (v4) -- (c12) ;

\end{tikzpicture}}
\caption{A connected Grundy coloring such that all the $c_j$'s are colored with color at least~$4$.}
\label{fig:connected-coloring-variables-clauses}
\end{figure}
}

The constant gadget $W$ is depicted in Figure~\ref{fig:connected-constant}.
The waves between $a_4$ and $a_6$ and between $a_9$ and $a_{11}$ correspond,  respectively, to the gadgets encoding the variables ($P_1$) and the clauses ($P_2$) described above and drawn in Figure~\ref{fig:connected-variables-clauses}.
A connected Grundy coloring achieving color~$7$ is given in Figure~\ref{fig:connected-coloring-constant} provided that going from $a_9$ to $a_{11}$ can be done without coloring any vertex $c_j$ with color~$2$ or less.

\begin{figure}[ht]
\centering
\scalebox{0.8}{\begin{tikzpicture}
\node[draw,circle,inner sep = -0.31cm] (a1) at (0,0) {$a_1$} ;
\node[draw,circle,inner sep = -0.31cm] (a2) at (1,0) {$a_2$} ;
\node[draw,circle,inner sep = -0.31cm] (a3) at (1,1) {$a_3$} ;
\node[draw,circle,double,inner sep = -0.31cm] (a4) at (2,0) {$a_4$} ;
\node[draw,circle,inner sep = -0.31cm] (a5) at (2,1) {$a_5$} ;
\node[draw,circle,inner sep = -0.31cm] (a6) at (3,1) {$a_6$} ;
\node[draw,circle,inner sep = -0.31cm] (a7) at (2,2) {$a_7$} ;
\node[draw,circle,inner sep = -0.31cm] (a8) at (3,0) {$a_8$} ;
\node[draw,circle,inner sep = -0.31cm] (a9) at (4,0) {$a_9$} ;
\node[draw,circle,inner sep = -0.34cm] (a10) at (4,1) {$a_{10}$} ;
\node[draw,circle,inner sep = -0.34cm] (a11) at (5,1) {$a_{11}$} ;
\node[draw,circle,inner sep = -0.34cm] (a12) at (3,2) {$a_{12}$} ;
\node[draw,circle,inner sep = -0.34cm] (a13) at (4,2) {$a_{13}$} ;
\node[draw,circle,inner sep = -0.34cm] (a14) at (5,2) {$a_{14}$} ;
\node[draw,circle,inner sep = -0.34cm] (a15) at (4,3) {$a_{15}$} ;
\node[draw,circle,inner sep = -0.34cm] (a16) at (6,1) {$a_{16}$} ;
\node[draw,circle,inner sep = -0.34cm] (a17) at (7,1) {$a_{17}$} ;
\node[draw,circle,inner sep = -0.34cm] (a18) at (8,0) {$a_{18}$} ;
\node[draw,circle,inner sep = -0.34cm] (a19) at (8,1) {$a_{19}$} ;
\node[draw,circle,inner sep = -0.34cm] (a20) at (8,2) {$a_{20}$} ;
\node[draw,circle,double,inner sep = -0.34cm] (a21) at (7,2) {$a_{21}$} ;
\node[draw,circle,inner sep = -0.34cm] (a22) at (7,3) {$a_{22}$} ;
\node[draw,circle,inner sep = -0.34cm] (a23) at (6,3) {$a_{23}$} ;
\node[draw,circle,double,inner sep = -0.34cm] (a24) at (9,1) {$a_{24}$} ;
\node[draw,circle,inner sep = -0.34cm] (a25) at (9,3) {$a_{25}$} ;
\node[draw,circle,inner sep = -0.34cm] (a26) at (8,3) {$a_{26}$} ;
\node[draw,circle,inner sep = -0.34cm] (a27) at (6,4) {$a_{27}$} ;

\node[draw,circle,inner sep = -0.34cm] (a28) at (5.3,5) {$a_{28}$} ;
\node[draw,circle,inner sep = -0.34cm] (a29) at (4.6,6) {$a_{29}$} ;
\node[draw,circle,inner sep = -0.34cm] (a30) at (5.3,7) {$a_{30}$} ;
\node[draw,circle,inner sep = -0.34cm] (a31) at (6.7,7) {$a_{31}$} ;
\node[draw,circle,inner sep = -0.34cm] (a32) at (7.4,6) {$a_{32}$} ;
\node[draw,circle,inner sep = -0.34cm] (a33) at (6.7,5) {$a_{33}$} ;

\draw (a28) -- (a29) -- (a30) -- (a31) -- (a32) -- (a33) -- (a28) -- (a30) -- (a32) -- (a28) -- (a31) -- (a33) -- (a29) -- (a31) ;
\draw (a29) -- (a32) ;
\draw (a30) -- (a33) ;

\draw (a28) -- (a27) -- (a33) ;

\draw (a27) -- (a15) -- (a7) -- (a3) -- (a2) -- (a1) -- (a3) -- (a4) -- (a5) -- (a7) -- (a6) -- (a8) -- (a9) -- (a10) -- (a12) -- (a13) -- (a14) -- (a16) -- (a17) -- (a18) -- (a19) -- (a21) -- (a23) ;
\draw (a11) -- (a12) -- (a15) -- (a13) ;
\draw (a14) -- (a15) ;
\draw (a27) -- (a22) -- (a20) ;
\draw (a16) -- (a23) -- (a17) ;
\draw (a23) -- (a27) -- (a26) -- (a25) -- (a27) ;
\draw (a19) -- (a20) -- (a21) ;
\draw (a18) -- (a24) -- (a25) ;
\draw[snake it] (a4) -- (a6) ;
\draw[snake it] (a9) -- (a11) ;

\draw[rounded corners] (a4) -- (2.6,-0.5) --(9,-0.5) -- (9.4,1) -- (a25) ;
\end{tikzpicture}}
\caption{The constant gadget. The doubly circled vertices are adjacent to all the $c_j$'s ($j \in [m]$). }
\label{fig:connected-constant}
\end{figure}
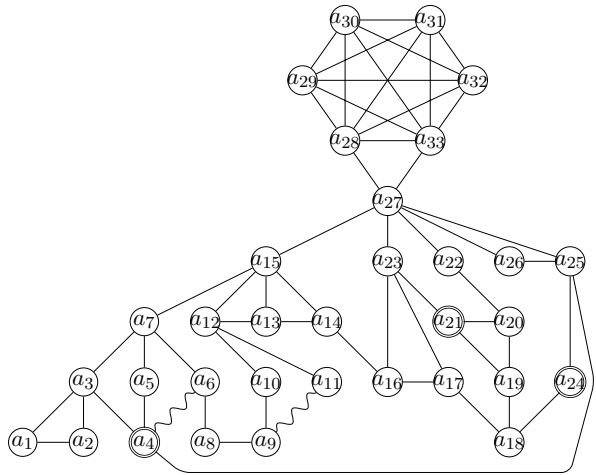

\begin{figure}[ht]
\centering

\scalebox{0.8}{\begin{tikzpicture}
\node[draw,circle,inner sep = -0.28cm] (a1) at (0,0) {$1$} ;
\node[draw,circle,inner sep = -0.28cm] (a2) at (1,0) {$2$} ;
\node[draw,circle,inner sep = -0.28cm] (a3) at (1,1) {$3$} ;
\node[draw,circle,double,inner sep = -0.28cm] (a4) at (2,0) {$1$} ;
\node[draw,circle,inner sep = -0.28cm] (a5) at (2,1) {$2$} ;
\node[draw,circle,inner sep = -0.28cm] (a6) at (3,1) {$1$} ;
\node[draw,circle,inner sep = -0.28cm] (a7) at (2,2) {$4$} ;
\node[draw,circle,inner sep = -0.28cm] (a8) at (3,0) {$2$} ;
\node[draw,circle,inner sep = -0.28cm] (a9) at (4,0) {$1$} ;
\node[draw,circle,inner sep = -0.28cm] (a10) at (4,1) {$2$} ;
\node[draw,circle,inner sep = -0.28cm] (a11) at (5,1) {$1$} ;
\node[draw,circle,inner sep = -0.28cm] (a12) at (3,2) {$3$} ;
\node[draw,circle,inner sep = -0.28cm] (a13) at (4,2) {$1$} ;
\node[draw,circle,inner sep = -0.28cm] (a14) at (5,2) {$2$} ;
\node[draw,circle,inner sep = -0.28cm] (a15) at (4,3) {$5$} ;
\node[draw,circle,inner sep = -0.28cm] (a16) at (6,1) {$1$} ;
\node[draw,circle,inner sep = -0.28cm] (a17) at (7,1) {$2$} ;
\node[draw,circle,inner sep = -0.28cm] (a18) at (8,0) {$1$} ;
\node[draw,circle,inner sep = -0.28cm] (a19) at (8,1) {$2$} ;
\node[draw,circle,inner sep = -0.28cm] (a20) at (8,2) {$1$} ;
\node[draw,circle,double,inner sep = -0.28cm] (a21) at (7,2) {$3$} ;
\node[draw,circle,inner sep = -0.28cm] (a22) at (7,3) {$2$} ;
\node[draw,circle,inner sep = -0.28cm] (a23) at (6,3) {$4$} ;
\node[draw,circle,double,inner sep = -0.28cm] (a24) at (9,1) {$2$} ;
\node[draw,circle,inner sep = -0.28cm] (a25) at (9,3) {$3$} ;
\node[draw,circle,inner sep = -0.28cm] (a26) at (8,3) {$1$} ;
\node[draw,circle,inner sep = -0.28cm] (a27) at (6,4) {$6$} ;

\node[draw,circle,inner sep = -0.28cm] (a28) at (5.3,5) {$1$} ;
\node[draw,circle,inner sep = -0.28cm] (a29) at (4.6,6) {$2$} ;
\node[draw,circle,inner sep = -0.28cm] (a30) at (5.3,7) {$3$} ;
\node[draw,circle,inner sep = -0.28cm] (a31) at (6.7,7) {$4$} ;
\node[draw,circle,inner sep = -0.28cm] (a32) at (7.4,6) {$5$} ;
\node[draw,circle,inner sep = -0.28cm] (a33) at (6.7,5) {$7$} ;

\draw (a28) -- (a29) -- (a30) -- (a31) -- (a32) -- (a33) -- (a28) -- (a30) -- (a32) -- (a28) -- (a31) -- (a33) -- (a29) -- (a31) ;
\draw (a29) -- (a32) ;
\draw (a30) -- (a33) ;

\draw (a28) -- (a27) -- (a33) ;

\draw (a27) -- (a15) -- (a7) -- (a3) -- (a2) -- (a1) -- (a3) -- (a4) -- (a5) -- (a7) -- (a6) -- (a8) -- (a9) -- (a10) -- (a12) -- (a13) -- (a14) -- (a16) -- (a17) -- (a18) -- (a19) -- (a21) -- (a23) ;
\draw (a11) -- (a12) -- (a15) -- (a13) ;
\draw (a14) -- (a15) ;
\draw (a27) -- (a22) -- (a20) ;
\draw (a16) -- (a23) -- (a17) ;
\draw (a23) -- (a27) -- (a26) -- (a25) -- (a27) ;
\draw (a19) -- (a20) -- (a21) ;
\draw (a18) -- (a24) -- (a25) ;
\draw[snake it] (a4) -- (a6) ;
\draw[snake it] (a9) -- (a11) ;

\draw[rounded corners] (a4) -- (2.6,-0.5) --(9,-0.5) -- (9.4,1) -- (a25) ;
\end{tikzpicture}}
\caption{A connected Grundy coloring of the constant gadget achieving color~$7$. The order is given by the sequence $(a_i)_{1 \leqslant i \leqslant 33}$.}
\label{fig:connected-coloring-constant}
\end{figure}
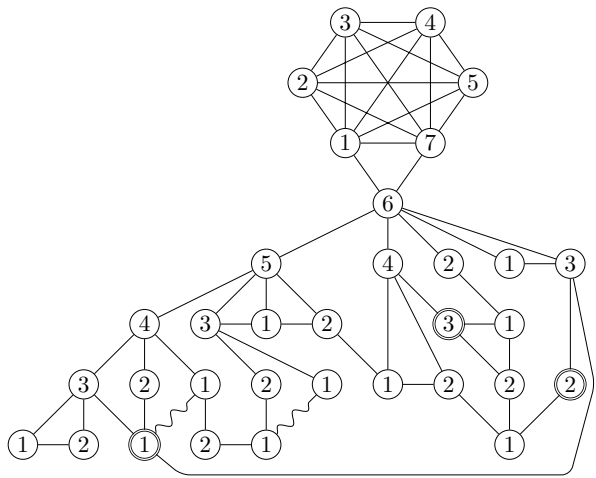

In the following claims, we use extensively Observation~\ref{obs:deg} which states that a vertex with degree~$d$ gets color at most~$d+1$.
We observe that coloring a vertex $z$ of degree~$d$ with color~$d+1$ is only useful if the ultimate goal is to achieve color~$d+1$.
Indeed, for $z$ to be colored with color $d+1$, all its neighbors have first to be colored (by each color from $1$ to $d$), which means that $z$ cannot be used thereafter.
Moreover, if one wants to color a neighbor $y$ of a vertex $x$ in order to color $x$ with a higher color, $y$ cannot receive a color greater than its degree $d(y)$.
Hence, the only vertices that could achieve color~$k$ are vertices of degree at least $k-1$ having at least one neighbor of degree at least $k-1$.

We call \emph{doubly circled vertices} the special vertices $a_4$, $a_{21}$ and $a_{24}$ (they are doubly circled in the figures).

\begin{foucClaim}\label{claim:a27}
To achieve color $7$, $a_{27}$ needs to be colored with color~$6$ (while for all $i \in [28,33]$, $a_i$ is still uncolored).
\end{foucClaim}

\appendixproof{Claim~\ref{claim:a27}}
{
\begin{proof}
One can achieve color~$7$ only in a vertex of degree at least~$6$ which has a neighbor of degree at least~$6$.
There are $m+7$ vertices of degree at least $6$: $a_{28}$ and $a_{33}$ (of degree 6), $a_{27}$ (of degree~$7$), all the $c_j$'s (of degree $8$), $v$ (of degree $2n+2$), $a_{24}$ (of degree $m+2$), $a_{21}$ (of degree $m+3$), and $a_4$ (of degree $m+4$).

As each vertex $c_j$ is adjacent to $a_4$, $a_{21}$ and $a_{24}$, we need to investigate the possibility of coloring with color~$7$, a vertex $c_j$, $a_4$, $a_{21}$, or $a_{24}$.
A vertex $c_j$ has two neighbors of degree~$2$ ($p_{3j-2}$ and $p_{3j}$; or $p_{3m-2}$ and $a_{11}$ in the special case of $c_m$), three neighbors of degree at most~$4$ (the three vertices corresponding to the literals of $C_j$) since a literal has at most two occurrences, and three vertices of degree more than $m+2$ ($a_4$, $a_{21}$, and $a_{24}$).
So, if no doubly circled vertex is colored yet, a vertex $c_j$ can be colored with a color at most~$5$. 
And if some doubly circled vertices are already colored but with always the same color, a vertex $c_j$ can be colored with a color at most~$6$ (when the shared color of the doubly circled vertices is~$5$). 

Let show that the three doubly circled vertices $a_4$, $a_{21}$, and $a_{24}$ cannot take two different colors both greater or equal to~$5$.
Indeed, suppose that two of those three vertices are colored with colors $p$ and $q$ such that $p<q$ and $p,q \geqslant 5$.
The doubly circled vertex colored with color~$q$ must have a vertex colored~$p$ in its neighborhood, but that color $p$ cannot come from a $c_j$ (since the vertex colored $p$ is adjacent to the $c_j$'s).
Thus, this color $p$ must come from another neighbor.
But, among all the neighbors of the doubly circled vertices which are not a vertex $c_j$, no vertex is of degree at least~$5$, a contradiction. 

From the last two paragraphs, we conclude that none of the vertices $a_4$, $a_{21}$, $a_{24}$, and the $c_j$'s can receive color~$7$.

The only other pairs of adjacent vertices both of degree at least~$6$ are the pairs of the triangle formed by $a_{27}$, $a_{28}$ and $a_{33}$.
We observe that $a_{27}$ is a cut-vertex whose removal disconnects the clique $K_6$ from the rest of the graph.
Hence, in a connected Grundy coloring, $a_{27}$ cannot get a color higher than~$6$ since its degree in one part of this cut is~$2$ and in the other part its degree is~$5$. 
Vertex $a_{33}$ (or by symmetry $a_{28}$) can be colored with color~$7$, but then $a_{27}$ has to be colored with color~$6$ otherwise it will lack a vertex colored~$6$ in its neighborhood.
The conclusion is that the only way to achieve color~$7$ is to color $a_{27}$ with color~$6$.
\end{proof}
}

\begin{foucClaim}\label{claim:first-line}
Vertices $a_{26}$, $a_{22}$, $a_{25}$, $a_{23}$, $a_{15}$ must receive color~$1$, $2$, $3$, $4$, $5$ respectively.
\end{foucClaim}
\appendixproof{Claim~\ref{claim:first-line}}
{
\begin{proof}
By Claim~\ref{claim:a27}, $a_{27}$ must be colored with color~$6$ before the clique $K_6$ is colored.
Thus, the five neighbors of $a_{27}$ which are not in the clique $K_6$ must get all the colors from~$1$ to~$5$.
Among those neighbors, the only vertex with degree~$5$ is $a_{15}$, so this vertex must get color~$5$.
Vertices $a_{23}$ and $a_{25}$ both have degree $4$ but for connectivity reasons $a_{26}$ cannot be colored before $a_{25}$, so $a_{25}$ cannot get a color higher than~$3$. 
Thus, $a_{23}$ must get color~$4$.
Vertex $a_{22}$ can bring a $1$ or a $2$ to $a_{27}$ while the pair $(a_{25},a_{26})$ can only bring the combinations $(1,2)$, $(2,1)$ or $(3,1)$.
Thus, the unique way to bring $1$, $2$ and $3$ to $a_{27}$ is that $a_{25}$ is colored~$3$, $a_{26}$ is colored~$1$, and $a_{22}$ is colored~$2$.
\end{proof}
}

\begin{foucClaim}\label{claim:a7}
Vertex $a_7$ must receive color~$4$.
\end{foucClaim}
\appendixproof{Claim~\ref{claim:a7}}
{
\begin{proof}
By Claim~\ref{claim:first-line},  $a_{15}$ has to receive color~$5$, so one of its four neighbors (apart from $a_{27}$) must receive color~$4$.
Only $a_7$ and $a_{12}$ have degree~$4$.
But $a_{12}$ cannot be colored~$4$ since its three neighbors $a_{10}$, $a_{11}$, and $a_{13}$ (apart from $a_{15}$) have only one neighbor which is neither $a_{12}$ nor $a_{15}$, so none of these vertices can bring color~$3$ to $a_{12}$.
\end{proof}
}

\begin{foucClaim}\label{claim:a3}
Vertex $a_3$ must receive color~$3$.
\end{foucClaim}
\appendixproof{Claim~\ref{claim:a3}}
{
\begin{proof}
By Claim~\ref{claim:a7},  $a_7$ must be colored~$4$.
Thus, one of its three neighbors $a_3$, $a_5$, and $a_6$ (apart from $a_{15}$) must receive color~$3$.
Vertices $a_3$ and $a_6$ have two neighbors apart from $a_7$.
But if $a_6$ is colored with color~$3$, then $a_4$ must be colored~$3$ to let colors~$1$ and~$2$ available for $a_3$ and $a_5$.
In that case, $a_3$ and $a_5$ would both receive color~$1$.
Another attempt is to color $a_1$ (or $a_2$) with~$1$, $a_3$ with $2$ but then $a_4$ has to be colored~$1$ and $a_5$ can no longer be colored~$1$.
Hence, only $a_3$ can be colored with~$3$.
\end{proof}
}

Claim~\ref{claim:a3} has further consequences: we must start the connected Grundy coloring by giving colors~$1$ and~$2$ to $a_1$ and $a_2$.
The only follow-up, for connectivity reasons, is then to color $a_3$ with color~$3$ and $a_4$ with color~$1$.
Thus, vertices $a_5$ and $a_6$ has to be colored with colors~$2$ and~$1$ respectively (so that $a_7$ can be colored~$4$).
As, by Claim~\ref{claim:first-line}, $a_{25}$ must receive color $3$, $a_{24}$ must receive color~$2$ (since $a_4$ has already color~$1$), so $a_{18}$ must be colored~$1$.

\begin{foucClaim}\label{claim:a21}
Vertex $a_{21}$ must receive color~$3$.
\end{foucClaim}
\appendixproof{Claim~\ref{claim:a21}}
{
\begin{proof}
By Claim~\ref{claim:first-line}, $a_{23}$ must get color~$4$, so its three neighbors apart from $a_{27}$ must receive colors~$1$, $2$ and~$3$. As $a_{20}$ must be colored~$1$ (in order to color $a_{22}$ with color~$2$), $a_{21}$ will be colored~$2$ or~$3$. 
Suppose $a_{21}$ is colored~$2$.
Then, $\{a_{16},a_{17}\}$ must be colored~$1$ and~$3$.
Vertex $a_{17}$ cannot be colored~$1$ since $a_{18}$ must get color~$1$, so $a_{16}$ must get color~$1$ and $a_{17}$, color~$3$.
In that case, $a_{17}$ lacks a vertex colored~$2$ in its neighborhood, and therefore cannot be colored~$3$.
So, $a_{21}$ has to be colored~$3$ and $a_{19}$ has to be colored~$2$ (since $a_{20}$ has to get color~$1$).
\end{proof}

A further consequence of Claim~\ref{claim:a21} is that $a_{16}$ must be colored~$2$ and $a_{17}$ must be colored~$1$ (the reverse being impossible, since $a_{18}$ has to be colored~$1$).
More importantly, we have now established that all the colored $c_j$'s (for each $j \in [m]$) have to be colored with color~$4$ or higher.
Indeed, we recall that the three doubly circled vertices (adjacent to all the $c_j$'s) $a_4$, $a_{21}$, and $a_{24}$ must respectively get color $1$, $3$, and~$2$. 
In particular, after having colored $a_1$ up to $a_4$, we cannot short-cut to $P_2$ since it will color a $c_j$ with~$2$, so we have to color $i_1$ with~$2$, $i_2$ with~$1$, and $v$ with~$2$.
As $v$ must be colored with color~$2$, none of the vertices encoding the literals can have color~$2$, so, again, we cannot short-cut from $P_1$ to $P_2$ otherwise, we would color a $c_j$ with~$2$.
Then, we can partly (or entirely) color $P_1$ but we have to color $a_6$ with~$1$, $a_8$ with $2$, and $a_9$ with~$1$. 
As $a_9$ is forced to get color $1$, $a_{10}$ has to give a $2$ to $a_{12}$ and $a_{11}$ is therefore forced to give color~$1$ to $a_{12}$.
}

\begin{foucClaim}\label{claim:a11}
The unique way of coloring $a_{11}$ with color~$1$ without coloring any vertex  $c_j$ with color~$1$, $2$, or~$3$ is to color all the $c_j$'s for each $j \in [m]$.
\end{foucClaim}
\appendixproof{Claim~\ref{claim:a11}}
{
\begin{proof}
We recall that the first four vertices to be colored are $a_1$, $a_2$, $a_3$, 
When going along the path from $a_9$ to $a_{11}$, the only vertex colored~$2$ which can be in the neighborhood of $c_j$ is $p_{3j-2}$. 
Indeed, we recall that the vertices encoding literals cannot be colored~$2$ since they are all adjacent to $v$ which is colored~$2$. 
By induction, as the only way to color vertex $p_{3j-2}$ with color~$2$ before $c_j$ is colored, is to color $c_{j-1}$, we have to color all the vertices in the path $P_2$. 
\end{proof}
}

We remark that opposite literals are adjacent, so for each $i \in [n]$, only one of $v_i$ and $\overline{v_i}$ can be colored with color~$3$.
We interpret coloring $v_i$ with~$3$ as setting $x_i$ to true and coloring $\overline{v_i}$ with~$3$ as setting $x_i$ to false.

\begin{foucClaim}\label{claim:cj}
To color each $c_j$ ($j \in [m]$) of the path $P_2$ with a color at least~$4$, the SAT formula must be satisfiable.  
\end{foucClaim}
\appendixproof{Claim~\ref{claim:cj}}
{
\begin{proof}
Each $c_j$ must have a vertex colored~$3$ in its neighborhood, but this vertex cannot be $a_{21}$ since this vertex cannot be colored yet. 
We recall that $a_{21}$ will be colored after $a_{11}$ is colored.
Thus, the vertex colored~$3$ can only belong to a set $\{v_i,\overline{v_i}\}$ encoding a literal $l_i$ such that $l_i$ is in $C_j$.
Indeed, the neighbors $p_{3j-2}$ and $p_{3j}$ are of degree~$2$ and $a_4$ is already colored~$1$.
Hence, there must be an assignment of the variables such that all the clauses of $\mathcal C$ are satisfied. 
As one cannot color both $v_i$ and $\overline{v_i}$ with color~$3$, the coloring of $P_1$ does constitute a feasible assignment.
\end{proof}
}

So, to achieve color~$7$ in a connected Grundy coloring, the \Sat formula must be satisfiable. The reverse direction consists of completing the coloring by giving $a_{13}$ color~$1$ and $a_{14}$ color~$2$, as shown in Figures~\ref{fig:connected-coloring-variables-clauses} and \ref{fig:connected-coloring-constant}.
\end{proof}

\section{Concluding remarks and questions}\label{sec:conclu}

To conclude this article, we suggest some questions which might be useful as a guide for further studies.

We have given two $O^*(c^n)$ exact algorithms for \PBG and \wGC with $c$ a constant, but we do not know whether such an algorithm exists for \cGC.

There is a gap between the $O^*(2^{O(wk)})$ algorithm of~\cite{TP97} and the lower bound of Theorem~\ref{thm:tw-LB-ETH}. Is \PBG in $\fpt$ when parameterized by the treewidth $w$? A simpler question is whether there is a better $O^*(f(k,w))$ algorithm (as noted in Observation~\ref{obs:tw-fpt}, if $f(k,w)=k^{O(w)}$ we directly obtain an $\fpt$ algorithm for parameter $w$). It could also be simpler to first determine whether \PBG is in $\fpt$ when parameterized by the feedback vertex set number (it is easy to see that it is in $\fpt$ when parameterized by the vertex cover number).

\PBG (parameterized by the number of colors) is in $\xp$, and we showed it to be in $\fpt$ on many important graph classes. Yet, the central question whether it is generally in $\fpt$ or $\wone$-hard remains unsolved. A perhaps more accessible research direction is to settle this question on bipartite graphs. 

\fullversion{It would also be interesting to determine the (classic) complexity of \PBG on interval graphs and chordal bipartite graphs (the latter question being asked in~\cite{thesis}).}\confversion{It would also be interesting to determine the (classic) complexity of \PBG on interval graphs.} Also, we saw that the algorithm of~\cite{TP97} implies a quasi-polynomial algorithm for planar (even apex-minor-free) graphs, making it unlikely to be $\np$-complete on this class. Is there a polynomial-time algorithm for such graphs?

We also recall that the exact polynomial-time approximation complexity of \PBG and \wGC is unknown; it is known that they admit no PTAS~\cite{GV97,Kortsarz}, but no $o(n)$-factor polynomial-time approximation algorithm is known. Recently, it was proved that  for any $r>1$, \PBG can be $r$-approximated in time $O^*(c^{n\log r/r})$ for some constant $c$, where $n$ is the graph's order~\cite{BLP15}. The approximation complexity of \cGC has not yet been studied.

Regarding \cGC, we showed that it remains $\np$-complete even for $k=7$. As \cGC is polynomial-time solvable for $k\leqslant 3$, its complexity status for $k=4,5,6$ remains open. It would also be interesting to study \cGC on restricted graph classes.

\fullversion{\paragraph{Acknowledgments.} We thank Sundar Vishwanathan for sending us the paper~\cite{GV97}.}

\bibliographystyle{plain}
\bibliography{biblio}


\end{document}